\documentclass[journal,10pt]{IEEEtran}

\normalsize
\usepackage{cite}
\ifCLASSINFOpdf
  \usepackage[pdftex]{graphicx}
\else
\fi
\usepackage[cmex10]{amsmath}
\usepackage{amsfonts}
\allowdisplaybreaks[4]
\usepackage{algorithm}
\usepackage{algorithmic}
\usepackage{array}
\usepackage{mdwmath}
\usepackage{mdwtab}
\usepackage{eqparbox}
\usepackage[tight,footnotesize]{subfigure}
\usepackage{amssymb}
\usepackage{amsthm}
\usepackage{multirow}

\begin{document}

\title{Joint Power Allocation and Beamforming for In-band Full-duplex Multi-cell Multi-user Networks}

\author{Haifeng Luo, 
        Navneet Garg, \IEEEmembership{Member, IEEE},
        Mark Holm,
        and Tharmalingam Ratnarajah, \IEEEmembership{Senior Member, IEEE}% <-this % stops a space
\thanks{H. Luo, N. Garg, and T. Ratnarajah are with Institute for Digital Communications, School of Engineering, The University of Edinburgh, UK, e-mails:\{ s1895225, navneet.garg, t.ratnarajah\}@ed.ac.uk }
\thanks{M. Holm is with Huawei Technologies (Sweden) AB, Gothenburg, Sweden, e-mail: mark.holm@huawei.com}
}

\markboth{IEEE Transactions on Vehicular Technology, Draft}
{}
%{Shell \MakeLowercase{\textit{et al.}}: Bare Demo of IEEEtran.cls for Journals}

% ----------------------------------------------------------------------------------------------------------------------------------------------------------
\maketitle

\begin{abstract}
This paper investigates a robust joint power allocation and beamforming scheme for in-band full-duplex multi-cell multi-user (IBFD-MCMU) networks. A mean-squared error (MSE) minimization problem is formulated with constraints on the power budgets and residual self-interference (RSI) power. The problem is not convex, so we decompose it into two sub-problems: interference management beamforming and power allocation, and give closed-form solutions to the sub-problems. Then we propose an iterative algorithm to yield an overall solution. The computational complexity and convergence behavior of the algorithm are analyzed. Our method can enhance the analog self-interference cancellation (ASIC) depth provided by the precoder with less effect on the downlink communication than the existing null-space projection method, inspiring a low-cost but efficient IBFD transceiver design. It can achieve $42.9\%$ of IBFD gain in terms of spectral efficiency with only antenna isolation, while this value increases to $60.9\%$ with further digital self-interference cancellation (DSIC). Numerical results illustrate that our algorithm is robust to hardware impairments and channel uncertainty. With sufficient ASIC depth, our method reduces the computation time by at least $20\%$ than the existing scheme due to its faster convergence speed at the cost of $<12.5\%$ sum rate loss. The benefit is much more significant with single-antenna users that our algorithm saves at least $40\%$ of the computation time at the cost of $<10\%$ sum rate reduction.
\end{abstract}
\vspace{-3mm}


\IEEEpeerreviewmaketitle

% ----------------------------------------------------------------------------------------------------------------------------------------------------------
\section{Introduction}
\label{sec:intro}
\IEEEPARstart{I}{n}-band full-duplex (IBFD) is a promising candidate in beyond fifth-generation (5G) systems due to its ability to enhance spectral efficiency (SE) and reduce the end-to-end transmission delay compared to conventional half-duplex (HD) systems since the bandwidth is used simultaneously by uplink and downlink transmissions. However, IBFD radios introduce additional interferences, i.e., self-interference (SI) and co-channel interference (CCI), degrading the promising gain of IBFD, especially the significant SI due to the proximity of the transceiver \cite{surveySic}. In IBFD multi-cell multi-user (MCMU) networks, downlink users not only get interference from base stations as in traditional HD radios but also get interference from uplink users. Meanwhile, the uplink communication is interfered by downlink transmissions. The complex interference must be appropriately processed to achieve substantial system throughput.

It is revealed that, in multi-input and multi-output (MIMO) systems, the interference is caused by the non-zero inner product between the channel and transmitted signal matrices. Therefore, beamforming can be leveraged to cancel the interference effectively. The readers are referred to \cite{overviewMimo} for details of implementing interference cancellation by adjusting the weights on different antennas from the perspective of signal processing. Interference cancellation in IBFD-MCMU networks is challenging due to the contradiction between a large amount of interference and the limited degrees of freedom imposed by finite numbers of antennas. Nevertheless, an appropriate beamforming design could be leveraged to effectively utilize the limited degrees of freedom and maximize spectral efficiency \cite{luo2023beam}. Authors of \cite{acMse} minimize the sum of mean-squared errors (MSE) of an IBFD multi-user network through beamforming, and authors of \cite{bfFd} maximize the sum rate of such a network by optimizing beamforming matrices, respectively. The results illustrate that appropriate beamforming schemes effectively manage the interference and improve the system capacity. IBFD-MCMU scenarios are further studied in \cite{paulaWsr} and a maximum weighted sum rate (MWSR) beamforming design with power constraints is formulated and solved by exploring the relationship between weighted minimum mean-squared error (WMMSE) and MWSR. The MWSR algorithm yields a joint power allocation and beamforming (JPABF) scheme that effectively manages the complex interference and maximizes the system throughput in IBFD-MCMU networks. Simulations demonstrate the spectral efficiency improvement of IBFD over HD with this design. However, existing studies usually assume self-interference cancellation (SIC) has been realized by other techniques rather than beamforming, which could be challenging in practice.

Indeed, various SIC techniques are proposed to suppress the SI in the propagation, analog, and digital domains. Studies demonstrate the effectiveness of these techniques by simulations and experiments as reported in the literature (see \cite{surveySic, surveyOptic} and references therein). However, enabling effective SIC in multi-input and multi-output (MIMO) systems is still challenging due to the exponentially increased cost and complexity with increasing antennas \cite{hlAccess,jzTwc}. For analog self-interference cancellation (ASIC), a total of $N_t N_r$ analog cancellers will be needed for an IBFD transceiver with $N_t$ transmitting antennas and $N_r$ receiving antennas. Although this number could be reduced to $N_r$ by using auxiliary transmitters, it is still a large number, and additional transmitter chains are costly \cite{hlChinacomm}. For nonlinear digital self-interference cancellation (DSIC), the parameters that need to be calculated become unacceptably high with high-order MIMO systems. It increases almost by the cube of $N_t$ with the nonlinearity order of 3 \cite{digitalSic}. Although the principal component analysis is leveraged to reduce the parameters by 65\% in \cite{digitalSic}, the computational complexity could still be unacceptable with massive MIMO (e.g., $>128$ antennas). Therefore, it is desired to explore beamforming for SIC since it does not require additional processing circuits nor increases the parameters. Existing studies in this direction mainly use the null-space projection (NSP) \cite{robertsBfc,tbSic, jointNsp, fdNsp} or choose orthogonal transmitting and receiving beamforming matrices \cite{robertsBfc}. The NSP-based methods steer the transmitted beams to the direction orthogonal to the SI channel, distorting the original precoders (i.e., transmitting beamformers \footnote{In the MIMO context, transmitting beamforming can be considered a form of precoding, especially when it comes to optimizing the signal considering the spatial aspect of the channel. Therefore, in our study, we generally use precoding to refer to transmitting beamforming.}) and introducing precoding errors. The trade-off between the precoding errors and SIC capability of the NSP-based methods is analyzed in \cite{tbSic}. It illustrates that achieving effective SIC may severely sacrifice the downlink capacity. Choosing orthogonal transmitting and receiving beamformers from the eigenvectors of the intended channel is a good approach to realizing both efficient SIC and effective downlink communications. However, it does not consider the limited dynamic range of practical receivers, which is an obstacle to all-digital SIC schemes \cite{hlChinacomm}. So it is not applicable to practical systems.

In addition to the beamforming, an appropriate power allocation policy can also manage the interference since the strength of interference depends on the transmit power of associated nodes. Power allocation could be implemented solely based on the channel strength in single-antenna systems as in \cite{fdResource}. However, in MIMO systems, power allocation and beamforming should be crossly designed for optimal performance since they will affect each other. Studies illustrate that an appropriate JPABF scheme can significantly improve the system capacity in MIMO systems. A sum rate maximization problem with respect to joint power allocation and beamforming is formulated for a single-cell 2-user network in \cite{xiao2018joint}. It is solved using an iterative algorithm, and simulation results show that it achieves close-to-bound sum rate performance. In addition to the system throughput, a JPABF scheme can also be designed to maximize the energy efficiency \cite{jointMu, Cirik2016beam}. The main challenge of the JPABF design is that the joint optimization problem is non-convex, so it is difficult to derive closed-form solutions \cite{jointMu, jointRelay}. Alternatively, it is solved by an iterative algorithm, causing high computational complexity.  

The JPABF has not yet been intensively studied for IBFD-MCMU networks, where a large amount of CCI and significant SI pose even greater challenges. To the best knowledge of the authors, existing studies in this direction mainly use the MWSR algorithm proposed in \cite{paulaWsr}. The MWSR can indeed achieve significant spectral efficiency; however, there are two disadvantages: 1) it requires sufficient ASIC to be realized by other techniques rather than precoding, which is challenging in practical MIMO systems; 2) the overall computational complexity could be high due to the slow convergence speed. In this paper, we propose a joint power allocation and interference management (JPAIM) algorithm for IBFD-MCMU networks, where the power allocation and beamforming are jointly optimized to minimize the sum of MSE of the network under constraints on the transmit power budget and residual self-interference (RSI). We do not assume a perfect SIC while we investigate the SIC capability of beamformers. Besides, we include practical imperfections (e.g., CSI errors, hardware impairments, and the limited dynamic range of receivers) to derive a robust solution. Compared to existing studies, our algorithm can achieve considerable IBFD gain with insufficient ASIC and reduce computational complexity. Our novel contributions can be summarised as follows.
\begin{itemize}
    \item Robustness to channel uncertainty: we include the effects of imperfect CSI by introducing channel uncertainty and derive a robust design. Simulation results show that our algorithm is robust to imperfect CSI as large channel uncertainty has a smaller impact on our method than on existing methods.
   \item Beamforming cancellation: we exploit beamforming for SIC and reveal that precoders can be leveraged to suppress the SI in the propagation domain, preventing receiver saturation. Thus, we enhance the ASIC capability of precoders to remove or reduce the requirements of RF cancellers, which have high implementation complexity and energy consumption. Different from existing NSP-based methods \cite{tbSic, jointNsp, fdNsp}, we realize it by adding a constraint on the received RSI power, yielding a joint design. Numerical results show the advantage of our design over existing methods that it can enhance the ASIC capability of precoders with less effect on the downlink communication. Thus, our method can enable a low-cost but efficient IBFD transceiver design.
   \item A new JPABF formulation: we use scalar coefficients to reflect the power allocation policy and formulate a minimization problem with respect to the coefficients and beamforming matrices. Since the formulated problem is not convex and may not be converted into a convex form by simple manipulations, we propose a two-stage method to obtain the solution. The sub-problems in each stage are convex, and closed-form solutions are given. Then we propose an iterative algorithm to obtain the overall solution as the optimized variables show inter-dependence on each other. We analyze the computational complexity and convergence behavior of the algorithm. Simulation results demonstrate that our algorithm significantly reduces the time complexity compared to the existing MWSR algorithm in \cite{paulaWsr} at the cost of acceptable sum rate loss due to the scalar power coefficients and decomposition. The benefit is significant with single-antenna users, which is meaningful for practical cellular network deployment.
\end{itemize}
The rest of the paper is organized as follows: In Section \ref{sec:preliminaries}, we give the system model of the IBFD-MCMU network with transceiver hardware impairments and channel uncertainty. Then, we formulate the joint optimization problem in Section \ref{sec:problem_formulation} and solve it by decomposing it into two sub-problems. The solutions of the two sub-problems are derived in Section \ref{sec:interferenceMng}, and the iterative algorithm to obtain the overall solution is given and analyzed. In Section \ref{sec:results}, we evaluate the performance of the proposed design by simulations. Finally, conclusions are drawn in Section \ref{sec:conclusions}. \\
\emph{Notations:}
$\mathbf{A}$, $\mathbf{a}$, and $a$ represent a matrix, a vector, and a scalar, respectively. $tr(\mathbf{A})$, $Cov(\mathbf{A})$, $\left| \mathbf{A} \right|$, $\left\| \mathbf{A} \right\|$, $\mathbf{A}^{H}$, $\mathbf{A}^{T}$, and $\mathbf{A}^{-1}$ denote the trace, covariance matrix, determinant, 2-norm, Hermitian, transpose, and inverse of matrix $\mathbf{A}$. $\mathbf{I}_{k}$ represents an identity matrix with $k$ elements along its diagonal. $\mathcal{D}(\mathbf{A})$ denotes a diagonal matrix containing the elements along the diagonal of $\mathbf{A}$. $[\mathbf{A}]_{i,j}$ represents the element of the $i^{th}$ row and $j^{th}$ column of matrix $\mathbf{A}$. $\mathbf{A}(:,j)$ denotes the $j^{th}$ column of matrix $\mathbf{A}$. $\mathbb{E}\left \{ \cdot \right \}$ denotes the expectation operation. $\max\left\{ \cdot \right\}$ and $\min\left\{ \cdot \right\}$ denote the maximum and minimum element of the set. $\mathcal{R} \left\{ \cdot \right\}$ represents the real part of complex numbers. $\mathcal{CN}(0,\sigma^2)$ denotes a complex normal distribution with zero mean and variance of $\sigma^2$.

% ----------------------------------------------------------------------------------------------------------------------------------------------------------
\section{System Model}
\label{sec:preliminaries}
We consider a $G$-cell network, where the base station in the $g^{th}$ cell serves  $K_g^d$ downlink (DL) users and $K_g^u$ uplink (UL) users. Assume all base stations (BS) have $N_{bs}$ transmitting antennas and $M_{bs}$ receiving antennas, while UL user equipment (UE) has $N_{ue}$ transmitting antennas and DL UE has $M_{ue}$ receiving antennas. Fig. \ref{fig:ibfd_mcmu} shows the interference between the nodes in this IBFD-MCMU network.
\begin{figure}[h]
      \centering
      \includegraphics[width=0.5\textwidth]{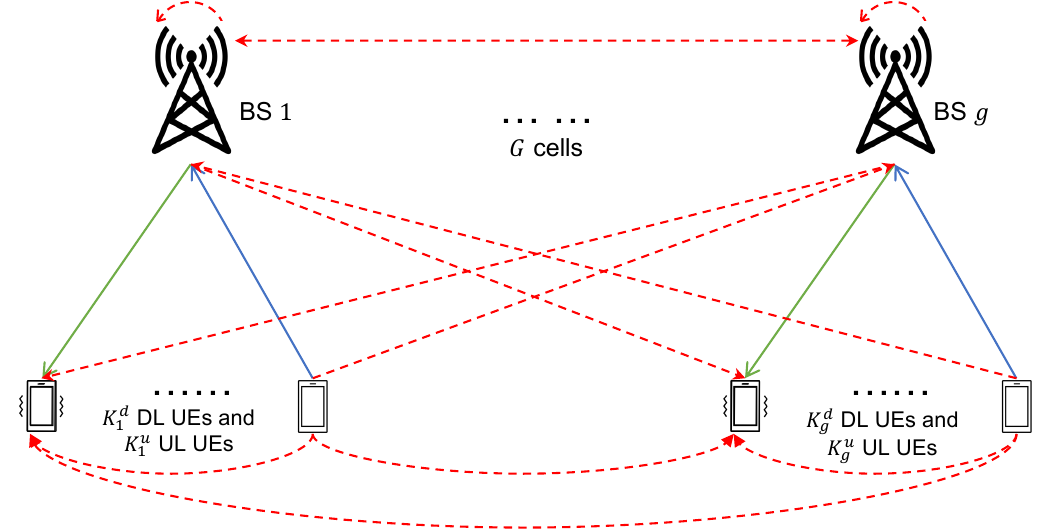}
      \caption{Interference between nodes in an IBFD multi-cell multi-user network.}
      \label{fig:ibfd_mcmu}
\end{figure}

\subsection{Transmitted Signals}
Let $\mathbf{s}_{k_g^d} \in \mathbb{C}^{b_d \times 1}$ represent the data symbols intended to the $k^{th}$ downlink user in the $g^{th}$ cell with statistics $\mathbb{E}\left\{  \mathbf{s}_{k_g^d}  \mathbf{s}_{k_g^d}^H \right\} = \mathbf{I}_{b_d}$; $\mathbf{V}_{k_g^d} \in \mathbb{C}^{N_{bs} \times b_d}$ denotes the associated precoding matrix; and $\alpha_{k_g^d}$ is the power coefficient that reflects the transmit power allocated to the corresponding user. The signal transmitted at the $g^{th}$ BS can be denoted as
\begin{equation}
\label{eq:xg}
    \mathbf{x}_g = \sum_{k=1}^{K_g^d} \left( \alpha_{k_g^d} \mathbf{V}_{k_g^d} \mathbf{s}_{k_g^d} + \mathbf{c}_{k_g^d} \right),
\end{equation}
where $ \mathbf{c}_{k_g^d}$ represents the hardware impairments due to the limited dynamic range of practical transmitters. The limited dynamic range is a natural consequence of imperfect digital-to-analog converters (DACs), oscillators, and power amplifiers (PAs). Experimental measurements demonstrate that the transmitter hardware impairments are independent of the transmitted signals and can be closely described by the circular complex Gaussian model as \cite{Day2012full}
\begin{equation}
    \mathbf{c}_g \sim \mathcal{CN} \left( \mathbf{0}, \kappa_{bs} \mathbb{E}\left \{ \mathcal{D}\left ( \mathbf{V}_g\mathbf{A}_g\mathbf{A}_g^H\mathbf{V}_g^H \right ) \right \} \right),
\end{equation}
where $\kappa_{bs}\ll 1$ characterizes the dynamic range of the transmitters at base stations. The base station transmits an accumulation of signals for DL users in the same cell. Let $\mathbf{s}_g = \left[ \mathbf{s}_{1_g^d}^T, \mathbf{s}_{2_g^d}^T, \dots, \mathbf{s}_{K_g^d}^T \right]^T$, $\mathbf{A}_g = \mathcal{D} \left( \alpha_{1_g^d} \mathbf{I}_{b_d}, \alpha_{2_g^d} \mathbf{I}_{b_d}, \dots, \alpha_{K_g^d} \mathbf{I}_{b_d}\right)$, and $\mathbf{V}_g = \left[ \mathbf{V}_{1_g^d}, \mathbf{V}_{2_g^d}, \dots, \mathbf{V}_{K_g^d} \right]$, the accumulated transmitted signals can be denoted as $\tilde{\mathbf{x}}_g = \mathbf{V}_g \mathbf{A}_g \mathbf{s}_g$. The transmit power of the $g^{th}$ base station is given as
\begin{equation}
\label{eq:pc_bs}
    \begin{split}
        P_g = \mathbb{E} \left \{ tr\left ( \tilde{\mathbf{x}}_g \tilde{\mathbf{x}}_g^H \right ) \right \} & = tr\left ( \mathbb{E}\left \{ \mathbf{V}_g\mathbf{A}_g\mathbf{s}_g \mathbf{s}_g^H\mathbf{A}_g^H\mathbf{V}_g^H \right \} \right ) \\
        & = \sum_{k=1}^{K_g^d} \alpha_{k_g^d}^2 tr\left ( \mathbf{V}_{k_g^d} \mathbf{V}_{k_g^d}^H \right ).
    \end{split}
\end{equation}
Similarly, let $\mathbf{s}_{k_g^u} \in \mathbb{C}^{b_u \times 1}$ denote the uplink payload symbols with statistics $\mathbb{E}\left\{  \mathbf{s}_{k_g^u}  \mathbf{s}_{k_g^u}^H \right\} = \mathbf{I}_{b_u}$; $\mathbf{V}_{k_g^u}$ represents the associated precoding matrix; $\gamma_{k_g^u}$ represents the power coefficient that reflects the allocated transmit power; and $\kappa_{ue}\ll 1$ characterizes the dynamic range of the transmitters at the user equipment. The transmitted signal of the $k^{th}$ uplink user in the $g^{th}$ cell (i.e., $k_g^u$) can be denoted as
\begin{equation}
    \mathbf{x}_{k_g^u} =  \underbrace{\gamma_{k_g^u} \mathbf{V}_{k_g^u} \mathbf{s}_{k_g^u}}_{\tilde{\mathbf{x}}_{k_g^u}} + \mathbf{c}_{k_g^u} ,
\end{equation}
where $\mathbf{c}_{k_g^u}$ denotes the transmitter hardware impairments at this uplink user that can be described as
\begin{equation}
    \mathbf{c}_{k_g^u} \sim \mathcal{CN} \left( \mathbf{0}, \kappa_{ue} \gamma_{k_g^u}^2 \mathbb{E} \left \{ \mathcal{D}\left ( \mathbf{V}_{k_g^u} \mathbf{V}_{k_g^u}^H \right ) \right \} \right) .
\end{equation}
The transmit power of the $k^{th}$ UE in the $g^{th}$ cell is given as
\begin{equation}
\label{eq:pc_ue}
    P_{k_g^u} = \mathbb{E} \left \{ tr\left ( \tilde{\mathbf{x}}_{k_g^u} \tilde{\mathbf{x}}_{k_g^u}^H \right ) \right \} = \gamma_{k_g^u}^2  tr\left ( \mathbf{V}_{k_g^u}\mathbf{V}_{k_g^u}^H \right ).
\end{equation}

\subsection{Received Signals}
Let $\mathbf{H}_{B,A}$ represent the coefficients matrix of the wireless MIMO channel from node $A$ to node $B$ throughout this paper, the signal received by the $g^{th}$ BS can be denoted as
\begin{equation}
\label{eq:received_signal_bs}
    \begin{split}
    \mathbf{y}_g & = \underbrace{\mathbf{H}_{g,k_g^u} \mathbf{x}_{k_g^u} + \mathop{\sum^{G}\sum^{K_i^u}}_{(i,j)\neq (k,g)} \mathbf{H}_{g,i_j^u} \mathbf{x}_{i_j^u} + \mathbf{H}_{g,g} \mathbf{x}_g +  \sum_{j\neq g}^{G} \mathbf{H}_{g,j} \mathbf{x}_j}_{\tilde{\mathbf{y}}_g} \\
    & \quad + \mathbf{n}_g + \mathbf{d}_g ,
    \end{split}
\end{equation}
where the third and fourth terms at the right hand represent the additional interference due to IBFD operation; $\mathbf{n}_g$ denotes the additive white Gaussian noise (AWGN) of the receiver such that $\mathbf{n}_g \sim \mathcal{CN} \left ( \mathbf{0}, \sigma_g^2\mathbf{I} \right )$; and $\mathbf{d}_g$ denotes the hardware impairments due to the limited dynamic range of practical receivers. The limited dynamic range is a natural consequence of imperfect low-noise amplifiers (LNAs), oscillators, and analog-to-digital converters (ADCs). Experimental measurements demonstrate that the receiver hardware impairments are independent of the received signals and can be closely described by the circular complex Gaussian model as \cite{Day2012full}
\begin{equation}
    \mathbf{d}_{g} \sim \mathcal{CN} \left( \mathbf{0}, \beta_{bs} \mathbb{E} \left \{ \mathcal{D}\left ( \tilde{\mathbf{y}}_g \tilde{\mathbf{y}}_g^H \right ) \right \} \right) ,
\end{equation}
where $\beta_{bs}\ll 1$ characterizes the dynamic range of the receivers at base stations. Similarly, the signal received by the $k^{th}$ downlink UE in the $g^{th}$ cell can be denoted as
\begin{equation}
\label{eq:received_signal_ue}
    \begin{split}
   \mathbf{y}_{k_g^d} &= \underbrace{\mathbf{H}_{k_g^d,g} \mathbf{x}_{k_g^d} + \mathop{\sum^{G}\sum^{K_i^d}}_{(i,j)\neq (k,g)} \mathbf{H}_{k_g^d,j} \mathbf{x}_{i_j^d} + \sum_{j=1}^{G} \sum_{i=1}^{K_j^u} \mathbf{H}_{k_g^d,i_j^u} \mathbf{x}_{i_j^u}}_{\tilde{\mathbf{y}}_{k_g^d}} \\
   & \quad + \mathbf{n}_{k_g^d} + \mathbf{d}_{k_g^d} ,
   \end{split}
\end{equation}
where the third term at the right hand represent the additional interference due to IBFD operation; $\mathbf{n}_{k_g^d}$ denotes the AWGN of the receiver such that $\mathbf{n}_{k_g^d} \sim \mathcal{CN} \left ( \mathbf{0}, \sigma_{k_g^d}^2\mathbf{I} \right )$; and $\mathbf{d}_{k_g^d}$ denotes the hardware impairments of the receiver described as
\begin{equation}
   \mathbf{d}_{k_g^d} \sim \mathcal{CN} \left( \mathbf{0}, \beta_{ue} \mathbb{E} \left \{ \mathcal{D}\left ( \tilde{\mathbf{y}}_{k_g^d} \tilde{\mathbf{y}}_{k_g^d}^H \right ) \right \} \right) ,
\end{equation}
where $\beta_{ue}\ll 1$ characterizes the dynamic range of the receivers at the user equipment. \\
\emph{Note}: The values of $\kappa_{bs}$, $\kappa_{ue}$, $\beta_{bs}$, and $\beta_{ue}$ are related to the measurable error vector magnitudes (EVMs) of corresponding RF transceivers. The HWIs model utilized is a verified model based on experiments \cite{Day2012full} and has been adopted by many studies in the field of wireless communications (see \cite{paulaWsr} and references therein). Due to the power consumption problem in MIMO systems, the transceivers tend to use low-resolution DACs/ADCs \cite{quantNoise, Choi2022energy}. So we assume that the dynamic range of receivers is mainly limited by DACs/ADCs, which can be described by the additive quantization noise model (AQNM) given in \cite{Choi2022energy}.

\vspace{-3mm}
\subsection{Channel Uncertainty}
Accurate channel estimation is challenging in practice due to limited training resources, resulting in channel uncertainty in the obtained channel state information (CSI). Let $\hat{\mathbf{H}}_{B,A}$ denote the estimate associated with the actual wireless channel $\mathbf{H}_{B,A}$, they are related as \cite{chError}
\begin{equation}
    \mathbf{H}_{B,A} = \hat{\mathbf{H}}_{B,A} + \mathbf{\Delta}_{B,A} ,
\end{equation}
where $\mathbf{\Delta}_{B,A}$ denotes the channel uncertainty (i.e., estimation errors). In this paper, we adapt the stochastic error model, which describes the channel uncertainty as $\mathbf{\Delta}_{B,A} \sim \mathcal{CN} \left( \mathbf{0}, \tilde{\sigma}_{B,A}^2 \mathbf{I} \right)$ \cite{paulaWsr, chError}. Using the accessible estimated CSI with the statistical channel uncertainty, the received signals given in Equations \eqref{eq:received_signal_bs} and \eqref{eq:received_signal_ue} can be written as
\begin{equation}
\label{eq:received_signal_bs_che}
    \begin{split}
    \mathbf{y}_g & = \hat{ \mathbf{H}}_{g,k_g^u} \mathbf{x}_{k_g^u} + \mathop{\sum^{G}\sum^{K_i^u}}_{i,j\neq k,g} \hat{\mathbf{H}}_{g,i_j^u} \mathbf{x}_{i_j^u} + \mathbf{H}_{g,g} \mathbf{x}_g \\
    & \quad +  \sum_{j\neq g}^{G} \hat{\mathbf{H}}_{g,j} \mathbf{x}_j + \mathbf{n}_g + \mathbf{d}_g + \mathbf{e}_g ,
     \end{split}
\end{equation}
\begin{equation}
\label{eq:received_signal_ue_che}
    \begin{split}
    \mathbf{y}_{k_g^d} & = \hat{\mathbf{H}}_{k_g^d,g} \mathbf{x}_{k_g^d} + \mathop{\sum^{G}\sum^{K_i^d}}_{i,j\neq k,g} \hat{\mathbf{H}}_{k_g^d,j} \mathbf{x}_{i_j^d} + \sum_{j=1}^{G} \sum_{i=1}^{K_j^u} \hat{\mathbf{H}}_{k_g^d,i_j^u} \mathbf{x}_{i_j^u} \\
    & \quad + \mathbf{n}_{k_g^d} + \mathbf{d}_{k_g^d} + \mathbf{e}_{k_g^d} ,
    \end{split}
\end{equation}
where $\mathbf{e}_g$ and $\mathbf{e}_{k_g^d}$ represent the noise due to imperfect CSI with statistics that $\mathbf{e}_g \sim \mathcal{CN} \left( \mathbf{0}, \hat{\sigma}_g^2 \mathbf{I}_{M_{bs}} \right)$ and $\mathbf{e}_{k_g^d} \sim \mathcal{CN} \left( \mathbf{0}, \hat{\sigma}_{k_g^d}^2 \mathbf{I}_{M_{ue}} \right)$; $\hat{\sigma}_g^2$ and $\hat{\sigma}_{k_g^d}^2$ are given as
\begin{align*}
    \hat{\sigma}_g^2 = \sum_{j=1}^{G} \sum_{i=1}^{K_j^u} \tilde{\sigma}_{g,i_j^u}^2 tr \left( \mathbf{T}_{i_j^u} \right) +  \sum_{j\neq g}^{G} \tilde{\sigma}_{g,j}^2 tr \left( \mathbf{T}_j \right) ,
\end{align*}
\begin{align*}
    \hat{\sigma}_{k_g^d}^2 = \sum_{j=1}^{G} \sum_{k=1}^{K_j^d} \tilde{\sigma}_{k_g^d,j}^2 tr \left( \mathbf{T}_{i_j^d} \right) + \sum_{j=1}^{G} \sum_{i=1}^{K_j^u} \tilde{\sigma}_{k_g^d,i_j^u}^2 tr \left( \mathbf{T}_{i_j^u} \right) ,
\end{align*}
as detailed in Appendix \ref{appendix:chError}. We use the perfect CSI of SI channels since it can be estimated with high SINR \cite{paulaWsr}.

% ----------------------------------------------------------------------------------------------------------------------------------------------------------
\section{Problem Formulation}
\label{sec:problem_formulation}
The power allocation and beamforming should have a common objective for the joint design. We chose the mean-squared error (MSE) minimization as the design objective, which is widely used in wireless networks \cite{jzTwc, acMse}.

\vspace{-3mm}
\subsection{Mean-Squared Errors}
Using the simplicities given in Appendix \ref{appendix:simp_mse}, the MSE of the $k^{th}$ downlink user and uplink user in the $g^{th}$ cell can be given as Equations \eqref{eq:mse_kgd} and \eqref{eq:mse_kgu} on the top of next page.
\newcounter{mytempeqncnt}
\begin{figure*}
    \normalsize
    \setcounter{mytempeqncnt}{\value{equation}}
    \begin{equation}
      \begin{split}
   \label{eq:mse_kgd}
     \varepsilon_{k_g^d}  \left ( \alpha, \gamma, \mathbf{V}, \mathbf{U} \right ) =  \mathbb{E} \left \{ \left \| \mathbf{s}_{k_g^d} - \mathbf{U}_{k_g^d}^H \mathbf{y}_{k_g^d} \right \|^2 \right \} & = tr \left ( \mathbf{U}_{k_g^d}^H \left ( \sum_{j=1}^{G} \sum_{i=1}^{K_j^d} \hat{\mathbf{H}}_{k_g^d,j} \mathbf{T}_{i_j^d} \hat{\mathbf{H}}_{k_g^d,j}^H + \sum_{j=1}^{G} \sum_{i=1}^{K_j^u} \hat{\mathbf{H}}_{k_g^d,i_j^u} \mathbf{T}_{i_j^u} \hat{\mathbf{H}}_{k_g^d,i_j^u}^H  \right )\mathbf{U}_{k_g^d}  \right )\\
    & \quad + \sigma_r^2 tr \left ( \mathbf{U}_{k_g^d}^H \mathcal{D} \left ( \sum_{j=1}^{G} \sum_{i=1}^{K_j^d} \hat{\mathbf{H}}_{k_g^d,j} \mathbf{T}_{i_j^d} \hat{\mathbf{H}}_{k_g^d,j}^H + \sum_{j=1}^{G} \sum_{i=1}^{K_j^u} \hat{\mathbf{H}}_{k_g^d,i_j^u} \mathbf{T}_{i_j^u} \hat{\mathbf{H}}_{k_g^d,i_j^u}^H  \right )\mathbf{U}_{k_g^d}  \right ) \\
    & \quad + \sigma_{k_g^d}^2 tr \left ( \mathbf{U}_{k_g^d}^H \mathbf{U}_{k_g^d} \right ) -2 \alpha_{k_g^d} \mathcal{R} \left \{ tr \left ( \mathbf{U}_{k_g^d}^H \hat{\mathbf{H}}_{k_g^d,g} \mathbf{V}_{k_g^d} \right ) \right \} + b_d ,
         \end{split}
  \end{equation}
    \hrulefill
    \vspace*{4pt}
\end{figure*}
\begin{figure*}
    \normalsize
    \setcounter{mytempeqncnt}{\value{equation}}
    \begin{equation}
      \begin{split}
   \label{eq:mse_kgu}
     \varepsilon_{k_g^u} \left ( \alpha, \gamma, \mathbf{V}, \mathbf{U} \right ) = \mathbb{E} \left \{ \left \| \mathbf{s}_{k_g^u} - \mathbf{U}_{k_g^u}^H \mathbf{y}_{g} \right \|^2 \right \} = & tr \left ( \mathbf{U}_{k_g^u}^H \left ( \sum_{j=1}^{G} \sum_{i=1}^{K_j^u} \hat{\mathbf{H}}_{g,i_j^u} \mathbf{T}_{i_j^u} \hat{\mathbf{H}}_{g,i_j^u}^H + \sum_{j=1}^{G} \hat{\mathbf{H}}_{g,j} \mathbf{T}_j \hat{\mathbf{H}}_{g,j}^H \right )\mathbf{U}_{k_g^u}  \right )\\
    & \quad + \sigma_r^2 tr \left (  \mathbf{U}_{k_g^u}^H \mathcal{D} \left ( \sum_{j=1}^{G} \sum_{i=1}^{K_j^u} \hat{\mathbf{H}}_{g,i_j^u} \mathbf{T}_{i_j^u} \hat{\mathbf{H}}_{g,i_j^u}^H + \sum_{j=1}^{G} \hat{\mathbf{H}}_{g,j} \mathbf{T}_j \hat{\mathbf{H}}_{g,j}^H \right )\mathbf{U}_{k_g^u}  \right ) \\
    & \quad + \sigma_{g}^2 tr \left ( \mathbf{U}_{k_g^u}^H \mathbf{U}_{k_g^u} \right ) -2 \gamma_{k_g^u} \mathcal{R} \left \{ tr \left ( \mathbf{U}_{k_g^u}^H \hat{\mathbf{H}}_{g,k_g^u} \mathbf{V}_{k_g^u} \right ) \right \} + b_u.
         \end{split}
  \end{equation}
    \hrulefill
    \vspace*{4pt}
\end{figure*}

\vspace{-3mm}
\subsection{Beamforming for Self-interference Cancellation}
\label{subsec:bf_asic}
We consider all-digital beamforming throughout this paper. The received RSI at the $g^{th}$ base station is given as
\begin{equation}
\label{eq:yg_si}
    \mathbf{y}_{g,\text{si}} = \mathbf{H}_{g,g} \sum_{k=1}^{K_g^d} \left ( \alpha_{k_g^d} \mathbf{V}_{k_g^d} \mathbf{s}_{k_g^d} + \mathbf{c}_{k_g^d}  \right ) ,
\end{equation}
which gives the power of the RSI as $ tr \left ( \mathbf{y}_{g,\text{si}} \mathbf{y}_{g,\text{si}} ^H \right )$. The RSI may exceed the dynamic range of receivers without active ASIC. Equation \eqref{eq:yg_si} indicates that the precoder can be leveraged to steer the beams to the desired direction in the propagation domain to suppress the SI as
\begin{equation}
      \begin{split}
    \min_{\mathbf{V}_{k_g^d}} \mathbb{E} \left \{ \left \| \mathbf{y}_{g,\text{si}} \right \|^2 \right \} = & \min_{\mathbf{V}_{k_g^d}} \alpha_{k_g^d}^2 \bigg ( tr \left ( \mathbf{H}_{g,g} \mathbf{V}_{k_g^d} \mathbf{V}_{k_g^d}^H \mathbf{H}_{g,g}^H \right ) \\
    & \quad + \kappa_{bs} tr \left ( \mathbf{H}_{g,g} \mathcal{D} \left ( \mathbf{V}_{k_g^d} \mathbf{V}_{k_g^d}^H \right ) \mathbf{H}_{g,g}^H \right ) \bigg ) .
     \end{split}
\end{equation}
It should be noted that the precoder suppresses the analog SI before the receiver to prevent receiver saturation. Thus the precoder-based cancellation is regarded as an ASIC technique, although it is implemented by digital signal processing. The receiving beamformers can also minimize the RSI as
\begin{equation}
    \min_{\mathbf{U}_{g}} \mathbb{E} \left \{ \left \| \mathbf{U}_{g}^H\mathbf{y}_{g,\text{si}} \right \|^2 \right \} .
\end{equation}
This process is taken after the ADCs of receivers, so it is regarded as a DSIC technique and does not help prevent receiver saturation. As previous studies illustrate, sufficient ASIC is critical to any effective SIC scheme for practical implementation, while active ASCI has high complexity and cost \cite{hlIcc, hlWcnc}. Therefore, it is desired to enhance the ASIC capability of precoders to implement a low-cost but effective IBFD transceiver. The RSI can be eliminated if $\mathbf{V}_{k_g^d}$ spans of vectors from the null-space of $\mathbf{H}_{g,g} \mathbf{H}_{g,g}^H + \kappa_{bs} \mathcal{D} \left ( \mathbf{H}_{g,g} \mathbf{H}_{g,g}^H \right )$. Existing NSP-based methods follow this idea \cite{robertsBfc,tbSic, jointNsp, fdNsp}, but they did not consider the downlink communication while doing the projection nor the transmitter distortion. Thus, the NSP may significantly reduce the downlink rate, and its performance is limited to transmitter distortions. Choosing orthogonal transmitting and receiving beamforming matrices as \cite{robertsBfc} can minimize $\mathbb{E} \left \{ \left \| \mathbf{U}_{g}^H\mathbf{y}_{g,\text{si}} \right \|^2 \right \}$, but it cannot minimize the RSI power before ADCs, i.e., $\mathbb{E} \left \{ \left \| \mathbf{y}_{g,\text{si}} \right \|^2 \right \}$. Thus, its performance is compromised by the limited dynamic range. \\
Let $l_g$ denote the pathloss of the effective SI channel at the $g^{th}$ base station, which includes the effects of other ASIC techniques (e.g., antenna isolation or RF cancellers), the ASIC capability of precoders (i.e., the ASIC depth provided by precoders) can be described in dB as
\begin{align}
\label{eq:asic_cap}
    \text{ASIC}_{g, \text{pre}} = 10 \times \log_{10} \frac{l_g \:  \mathbb{E}\left \{ tr\left ( \mathbf{x}_{g} \mathbf{x}_{g}^H \right ) \right \}}{ \mathbb{E} \left \{ tr\left ( \mathbf{y}_{g,\text{si}} \mathbf{y}_{g,\text{si}}^H \right ) \right \}},
\end{align}
where the numerator describes the received power without precoding, and the denominator is the RSI power with precoding; $ \mathbf{x}_{g}$ and $\mathbf{y}_{g,\text{si}}$ are given in \eqref{eq:xg} and \eqref{eq:yg_si}, respectively.

\vspace{-3mm}
\subsection{Constrained Optimization Problem}
Minimizing the sum of MSE by optimizing precoders will include the effects of combiners, as the MSE expressions suggest. As a result, precoders may not provide a sufficient ASIC depth to prevent receiver saturation if combiners have effectively suppressed the RSI and distortions caused by it. However, the degrees of freedom of combiners occupied by SIC and uncorrelated distortions due to receiver saturation will naturally compromise the uplink communication. Thus, the existing beamforming schemes (e.g., MWSR in \cite{paulaWsr}, and other schemes in \cite{Cirik2016beam,bfFd,jzTwc}) for IBFD networks will require sufficient ASIC realized by other techniques to guarantee their performance. In contrast, we would like the precoder to provide a sufficient ASIC depth to reduce the RSI power, which can be done by adding the constraint on the RSI power, formulating the problem as
\begin{align}
\label{eq:opt_problem0}
        (P.1) \; \;  & \min_{\left\{ \alpha, \gamma, \mathbf{V}, \mathbf{U} \right\}_{\forall k,g}} \sum_{g=1}^{G}\sum_{k=1}^{K_g^d} \varepsilon_{k_g^d} \left ( \alpha, \gamma, \mathbf{V}, \mathbf{U} \right ) \nonumber \\
        & \;\;\;\;\;\;\;\;\;\;\; + \sum_{g=1}^{G}\sum_{k=1}^{K_g^u} \varepsilon_{k_g^u}  \left ( \alpha, \gamma, \mathbf{V}, \mathbf{U} \right ) \\
                        & \text{s.t.} \; \;  (C.1) \;\; \sum_{k=1}^{K_g^d} \alpha_{k_g^d}^2  tr\left ( \mathbf{V}_{k_g^d} \mathbf{V}_{k_g^d}^H \right ) \leq P_{bs}, \; \forall \; g  \\
                        & \;\;\;\;\;\; (C.2) \;\; \gamma_{k_g^u}^2  tr\left ( \mathbf{V}_{k_g^u} \mathbf{V}_{k_g^u}^H \right ) \leq P_{ue}, \; \forall \; k,g    \\
                        & \;\;\;\;\;\;  (C.3) \;\; \epsilon_{\text{rsi}, g}\left ( \mathbf{V}\right ) \leq \bar{\epsilon}_{\text{rsi},g} ,
\end{align}
where $\epsilon_{\text{rsi}, g}\left ( \mathbf{V}\right ) = \mathbb{E} \left \{ tr\left ( \mathbf{y}_{g,\text{si}} \mathbf{y}_{g,\text{si}}^H \right ) \right \}$ represents the received analog RSI power; $\bar{\epsilon}_{\text{rsi},g}$ is the tolerable RSI power that could be chosen according to the dynamic range of receivers. We want to suppress the SI solely by precoding schemes instead of reducing the transmit power of the downlink payload. Thus, the RSI is a function of only precoders, as denoted in the formulation. According to the weighted sum rate and WMMSE relationship in \cite{paulaWsr}, minimizing the sum MSE under the RSI power constraint (i.e., problem $(P.1)$) can maximize the sum rate of the network with a constrained RSI power (i.e., enhanced ASIC depth to prevent receiver saturation).
\newtheorem{lemma}{Lemma}
\begin{lemma}
\label{lemma:opt_problem}
    The optimization problem $(P.1)$ is equivalent to the problem $(P.2)$ shown below with appropriate $\nu_g$ . \\
    \begin{align}
    \label{eq:opt_problem}
            (P.2) \; \; & \min_{\left\{ \alpha, \gamma, \mathbf{V}, \mathbf{U} \right\}_{\forall k,g}}  \sum_{g=1}^{G}\sum_{k=1}^{K_g^d} \varepsilon_{k_g^d} \left ( \alpha, \gamma, \mathbf{V}, \mathbf{U} \right ) \nonumber \\
            & \quad + \sum_{g=1}^{G}\sum_{k=1}^{K_g^u} \varepsilon_{k_g^u}  \left ( \alpha, \gamma, \mathbf{V}, \mathbf{U} \right )  + \sum_{g=1}^{G} \nu_g \epsilon_{\text{rsi}, g}  \left ( \mathbf{V}\right ) \\
                        \label{eq:opt_c1}
                        & \text{s.t.} \; \;  (C.1) \;\; \sum_{k=1}^{K_g^d} \alpha_{k_g^d}^2  tr\left ( \mathbf{V}_{k_g^d} \mathbf{V}_{k_g^d}^H \right ) \leq P_{bs}, \; \forall \; g  \\
                        \label{eq:opt_c2}
                        & \;\;\;\;\;\;  (C.2) \;\; \gamma_{k_g^u}^2  tr\left ( \mathbf{V}_{k_g^u} \mathbf{V}_{k_g^u}^H \right ) \leq P_{ue}, \; \forall \; k,g ,
    \end{align}
\end{lemma}
\begin{proof}
See Appendix \ref{appendix:lemma_proof}.
\end{proof}
According to Lemma \ref{lemma:opt_problem}, the optimal power allocation and beamforming schemes can be obtained from solving problem $(P.2)$. Still, an appropriate value of $\nu_g$ needs to be chosen at first to formulate the problem. As stated in Appendix \ref{appendix:lemma_proof}, $\nu_g$ should be equal to $\varpi_{\text{rsi}, g}$ to make the problems equivalent, while $\varpi_{\text{rsi}, g}$ depends on the chosen tolerable RSI power $\bar{\epsilon}_{\text{rsi},g}$. It is challenging to mathematically derive the value of $\nu_g$ from given $\bar{\epsilon}_{\text{rsi},g}$ since the multiple variables are entangled. However, we can directly give the value of $\nu_g$, and it will reflect a specific tolerable RSI power. The resulting solutions will maximize the sum rate under the corresponding RSI power constraint. We choose $\nu_g$ according to its properties. As \eqref{eq:opt_problem} suggests, $\nu_g$ can adjust the preference of the precoder: the objective function will be dominated by the sum of MSE when $\nu_g$ is relatively small, and conversely, the objective function will be dominated by the RSI when $\nu_g$ is rather large. As a result, precoders will tend to minimize the sum of MSE or RSI accordingly. Therefore, we can set a relatively large $\nu_g$ to help suppress SI before the receiver if there is not sufficient ASIC depth and set a relatively small $\nu_g$ to minimize the precoding errors if there is sufficient ASIC depth. According to experiments, we can set the value of $\nu_g$ based on the realized ASIC depth $l_g$ as $\nu_g = 10^{-l_{g}/5}$. \\
From the MSE expressions, we can see that the main difficulty of the minimization problem $(P.2)$ is that the optimized variables are entangled, leading to a non-convex problem. The non-convex constrained optimization problem may not be converted to a convex form with simple manipulations. Thus, we propose a two-stage approach to solve it, where the original joint optimization problem is decomposed into two convex sub-problems: interference management through beamforming and power allocation. The two sub-problems are relatively easy to be solved through standard solutions, but the decomposition leads to a suboptimal solution to the joint optimization problem.

% ----------------------------------------------------------------------------------------------------------------------------------------------------------
\section{Joint Power Allocation and Interference Management Algorithm}
\label{sec:interferenceMng}
In the interference management stage, we will fix the power coefficients to form the first sub-problem regarding beamforming matrix optimization. However, the objective function in \eqref{eq:opt_problem} is still not jointly convex to the transmitting and receiving beamforming matrices (i.e., precoders and combiners). Thus, we further divide this sub-problem into combiner updating and precoder updating stages.

\vspace{-2mm}
\subsection{Receiving Beamforming}
With fixed precoding matrices and power coefficients, the combining matrices are optimized to minimize the sum of MSE. The analog RSI power is not affected by the combiners, so it is removed from the objective function, and the combiner optimization problem is formulated as
\begin{equation}
\label{eq:opt_subproblem1.1}
      \begin{split}
        (S.1.1) \; \min_{\left\{ \mathbf{U} \right\}_{\forall k,g}} & \sum_{g=1}^{G}\sum_{k=1}^{K_g^d} \varepsilon_{k_g^d}  \left ( \alpha, \gamma, \mathbf{V}, \mathbf{U} \right )  \\
        & \quad + \sum_{g=1}^{G}\sum_{k=1}^{K_g^u} \varepsilon_{k_g^u}  \left ( \alpha, \gamma, \mathbf{V}, \mathbf{U} \right ).
     \end{split}
\end{equation}
The objective function (i.e., the sum of MSE) is convex and differentiable to the combining matrices with other variables fixed, and there is no constraint associated with combining matrices. Thus, we can differentiate the objective function with respect to $\mathbf{U}_{k_g^d}$ and $\mathbf{U}_{k_g^u}$, and set the derivatives to zero, then the optimal combining matrices for problem $(S.1.1) $ in \eqref{eq:opt_subproblem1.1} are given as
\begin{align}
\label{eq:slt_combiner_dl}
    \mathbf{U}_{k_g^d}^{*}  & = \alpha_{k_g^d} \mathbf{C}_{k_g^d}^{-1} \hat{\mathbf{H}}_{k_g^d,g} \mathbf{V}_{k_g^d} , \\
\label{eq:slt_combiner_ul}
    \mathbf{U}_{k_g^u}^{*} & = \gamma_{k_g^u} \mathbf{C}_{g}^{-1} \hat{\mathbf{H}}_{g,k_g^u} \mathbf{V}_{k_g^u} ,
\end{align}
where $\mathbf{C}_{k_g^d}$ and $\mathbf{C}_{g}$ are covariance matrices of the received signals at the downlink user $k_g^d$ and the $g^{th}$ base station given in Equations \eqref{eq:Ckgd} and \eqref{eq:Cg} in Appendix \ref{appendix:simp_mse}.

\vspace{-3mm}
\subsection{Transmitting Beamforming}
With fixed combining matrices and power coefficients, the sub-problem regarding precoder optimization is formulated as
\begin{align}
\label{eq:opt_subproblem1.2}
         (S.1.2) & \; \min_{\left\{  \mathbf{V} \right\}_{\forall k,g}} \sum_{k=1}^{K_g^d} \varepsilon_{k_g^d}  \left ( \alpha, \gamma, \mathbf{V}, \mathbf{U} \right ) +  \sum_{g=1}^{G} \nu_g \epsilon_{\text{rsi}, g}  \left ( \mathbf{V}\right ) \nonumber \\
         & \quad + \sum_{g=1}^{G}\sum_{k=1}^{K_g^u} \varepsilon_{k_g^u}  \left ( \alpha, \gamma, \mathbf{V}, \mathbf{U} \right )   \\
                        \label{eq:s1opt_c1}
                        \text{s.t.} \; \; & (C.1) \;\; \sum_{k=1}^{K_g^d} \alpha_{k_g^d}^2  tr\left ( \mathbf{V}_{k_g^d} \mathbf{V}_{k_g^d}^H \right ) \leq P_{bs}, \; \forall \; g  \\
                        \label{eq:s1opt_c2}
                        & (C.2) \;\; \gamma_{k_g^u}^2  tr\left ( \mathbf{V}_{k_g^u} \mathbf{V}_{k_g^u}^H \right ) \leq P_{ue}, \; \forall \; k,g.
\end{align}
To solve the constrained optimization problem, we need to augment the objective function with a weighted sum of the constraint functions \cite{convexOpt}, yielding the Lagrange function as Equation \eqref{eq:opt_lagrange_s1.2}, where $\varpi_{g} \geq 0$ and $ \varpi_{k_g^u} \geq 0$ are Lagrange multipliers associated with constraints \eqref{eq:s1opt_c1} and \eqref{eq:s1opt_c2}. 
\begin{figure*}
    \normalsize
    \setcounter{mytempeqncnt}{\value{equation}}
    \begin{equation}
    \label{eq:opt_lagrange_s1.2}
      \begin{split}
     \mathcal{L} & \left ( \mathbf{V}, \nu, \varpi \right )  = \sum_{g=1}^{G}\sum_{k=1}^{K_g^d} \varepsilon_{k_g^d}  \left ( \alpha, \gamma, \mathbf{V}, \mathbf{U} \right ) + \sum_{g=1}^{G}\sum_{k=1}^{K_g^u} \varepsilon_{k_g^u} \left ( \alpha, \gamma, \mathbf{V}, \mathbf{U} \right ) + \sum_{g=1}^{G}  \varpi_{g} \left ( \sum_{k=1}^{K_g^d} \alpha_{k_g^d}^2 tr \left ( \mathbf{V}_{k_g^d} \mathbf{V}_{k_g^d}^H \right ) - P_{bs} \right )  \\
   & \quad + \sum_{g=1}^{G} \sum_{k=1}^{K_g^u} \varpi_{k_g^u} \left ( \gamma_{k_g^u}^2 tr \left ( \mathbf{V}_{k_g^u} \mathbf{V}_{k_g^u}^H \right ) - P_{ue} \right ) +  \sum_{g=1}^{G} \nu_g \sum_{k=1}^{K_g^d} \alpha_{k_g^d}^2 tr\left ( \mathbf{H}_{g,g} \mathbf{V}_{k_g^d} \mathbf{V}_{k_g^d}^H \mathbf{H}_{g,g}^H + \kappa_{bs}\mathbf{H}_{g,g} \mathcal{D}\left ( \mathbf{V}_{k_g^d} \mathbf{V}_{k_g^d}^H \right ) \mathbf{H}_{g,g}^H \right ) 
         \end{split}
  \end{equation}
    \hrulefill
    \vspace*{4pt}
\end{figure*}
Differentiate the Lagrange function \eqref{eq:opt_lagrange_s1.2} with respect to $\mathbf{V}_{k_g^d}$ and $\mathbf{V}_{k_g^u}$, and set the derivatives to zero, the optimal precoding matrices for the sub-problem $(S.1.2)$ are given as
\begin{equation}
\label{eq:slt_precoder_dl}
    \begin{split}
   \mathbf{V}_{k_g^d}^{*} & =  \frac{1}{\alpha_{k_g^d}} \Big ( \mathbf{\Omega}_{g} + \nu_g \left ( \mathbf{H}_{g,g}^H \mathbf{H}_{g,g} + \kappa_{bs} \mathcal{D} \left ( \mathbf{H}_{g,g}^H \mathbf{H}_{g,g} \right ) \right )  \\
   & \quad + \varpi_{g}^{*} \mathbf{I} \Big )^{-1} \mathbf{H}_{k_g^d,g}^H \mathbf{U}_{k_g^d} ,
   \end{split}
\end{equation}
\begin{align}
\label{eq:slt_precoder_ul}
   \mathbf{V}_{k_g^u}^{*} =  \frac{1}{\gamma_{k_g^d}} \left ( \mathbf{\Omega}_{k_g^u} + \varpi_{k_g^u}^{*} \mathbf{I} \right )^{-1} \mathbf{H}_{g,k_g^u}^H \mathbf{U}_{k_g^u} .
\end{align}
where $\mathbf{\Omega}_{k_g^d}$ and $\mathbf{\Omega}_{k_g^u}$ represent the interference-plus-distortions matrices for associated users, which can be denoted using the function $\mathcal{F}_1 (\cdot)$ defined in Appendix \ref{appendix:functions} as
\begin{align*}
    \mathbf{\Omega}_{g} & = \sum_{j=1}^{G} \sum_{i=1}^{K_j^d} \mathcal{F}_1 ( \hat{\mathbf{H}}_{i_g^d,g}^H, \mathbf{U}_{i_j^d} ) +\sum_{j=1}^{G} \sum_{i=1}^{K_j^u} \mathcal{F}_1 ( \hat{\mathbf{H}}_{j,g}^H, \mathbf{U}_{i_j^u} )  , \\
    \mathbf{\Omega}_{k_g^u} & = \sum_{j=1}^{G} \sum_{i=1}^{K_j^d}  \mathcal{F}_1 ( \hat{\mathbf{H}}_{i_g^d,k_g^u}^H, \mathbf{U}_{i_j^d} ) + \sum_{j=1}^{G} \sum_{i=1}^{K_j^u} \mathcal{F}_1 ( \hat{\mathbf{H}}_{j,k_g^u}^H, \mathbf{U}_{i_j^u} ) .
\end{align*}
Then, we need to find $\varpi_{g}^{*}\geq 0$ and $\varpi_{k_g^u}^{*} \geq 0$ satisfying power constraints \eqref{eq:s1opt_c1} and \eqref{eq:s1opt_c2}, which could be found through Bisection searching. However, calculating the precoder and its norm with different values of $\varpi_{g}$ and $\varpi_{k_g^u}$ for each iteration during the Bisection searching will yield high computational complexity. To reduce the complexity, we use singular value decomposition (SVD) to convert the expressions of the transmit power into a scalar form. Performing SVD to $ \mathbf{\Omega}_{g} + \nu \left ( \mathbf{H}_{g,g}^H \mathbf{H}_{g,g} + \kappa_{bs} \mathcal{D} \left ( \mathbf{H}_{g,g}^H \mathbf{H}_{g,g} \right ) \right ) $, which is Hermitian, we have $ \mathbf{\Omega}_{g} + \nu \left ( \mathbf{H}_{g,g}^H \mathbf{H}_{g,g} + \kappa_{bs} \mathcal{D} \left ( \mathbf{H}_{g,g}^H \mathbf{H}_{g,g} \right ) \right )  = \mathbf{Q}_{k_g^d} \mathbf{D}_{k_g^d} \mathbf{Q}_{k_g^d}^H$, where $\mathbf{Q}_{k_g^d} \mathbf{Q}_{k_g^d}^H = \mathbf{I}$. Then, Equation \eqref{eq:slt_precoder_dl} can be written as
\begin{align*}
    \mathbf{V}_{k_g^d}^{*} & =  \frac{1}{\alpha_{k_g^d}} \left ( \mathbf{Q}_{k_g^d} \mathbf{D}_{k_g^d} \mathbf{Q}_{k_g^d}^H + \mathbf{Q}_{k_g^d} \varpi_{g}^{*} \mathbf{Q}_{k_g^d}^H \right )^{-1} \mathbf{H}_{k_g^d,g}^H \mathbf{U}_{k_g^d} \\
    &  = \frac{1}{\alpha_{k_g^d}} \mathbf{Q}_{k_g^d} \left ( \mathbf{D}_{k_g^d} + \varpi_{g}^{*} \mathbf{I} \right )^{-1} \mathbf{Q}_{k_g^d}^H \mathbf{H}_{k_g^d,g}^H \mathbf{U}_{k_g^d}.
\end{align*}
Let $\mathbf{G}_{k_g^d} = \mathbf{Q}_{k_g^d}^H \mathbf{H}_{k_g^d,g}^H \mathbf{U}_{k_g^d} \mathbf{U}_{k_g^d}^H \mathbf{H}_{k_g^d,g} \mathbf{Q}_{k_g^d}$, the power constraint in Equation \eqref{eq:s1opt_c1} can be written as
\begin{align*}
    \sum_{k=1}^{K_g^d} \alpha_{k_g^d}^2 tr\left ( \mathbf{V}_{k_g^d}^{*} (\mathbf{V}_{k_g^d}^{*})^H \right )  = \sum_{k=1}^{K_g^d} \sum_{n=1}^{N_{bs}} \frac{[ \mathbf{G}_{k_g^d} ]_{n,n} }{ \left ( [ \mathbf{D}_{k_g^d} ]_{n,n} + \varpi_{g} \right )^2} \leq P_{bs} 
\end{align*}
Similarly, performing SVD to $\mathbf{\Omega}_{k_g^u}$ such that $\mathbf{\Omega}_{k_g^u} = \mathbf{Q}_{k_g^u} \mathbf{D}_{k_g^u} \mathbf{Q}_{k_g^u}^H$, and let $\mathbf{G}_{k_g^u} = \mathbf{Q}_{k_g^u}^H \mathbf{H}_{g, k_g^u}^H \mathbf{U}_{k_g^u} \mathbf{U}_{k_g^u}^H \mathbf{H}_{g,k_g^u} \mathbf{Q}_{k_g^u}$. The power constraint in Equation \eqref{eq:s1opt_c2} can be written as
\begin{equation}
     \sum_{n=1}^{N_{ue}} \frac{[ \mathbf{G}_{k_g^u} ]_{n,n} }{ \left ( [ \mathbf{D}_{k_g^u} ]_{n,n} + \varpi_{k_g^u} \right )^2} \leq P_{ue} .
\end{equation}
Since $\mathbf{G}_{k_g^d}$ and $\mathbf{D}_{k_g^d}$ (or $\mathbf{G}_{k_g^u}$ and $\mathbf{D}_{k_g^u}$) are fixed while updating the precoding matrices, we only need to calculate their values once before the Bisection searching, then only scalar calculations are required during the process of Bisection searching.

% ----------------------------------------------------------------------
\subsection{Power Allocation}
In this section, we will fix the beamforming matrices and derive the solutions to the power coefficients. The sub-problem with regard to power allocation policy optimization is formulated as
\begin{align}
\label{eq:opt_subproblem2}
         (S.2) & \; \min_{ \left\{  \alpha, \gamma \right\}_{\forall k,g}} \sum_{k=1}^{K_g^d} \varepsilon_{k_g^d}  \left ( \alpha, \gamma, \mathbf{V}, \mathbf{U} \right )  + \sum_{g=1}^{G}\sum_{k=1}^{K_g^u} \varepsilon_{k_g^u}  \left ( \alpha, \gamma, \mathbf{V}, \mathbf{U} \right )  \\
                        \label{eq:sopt_c1}
                        & \text{s.t.} \; \;  (C.1) \;\; \sum_{k=1}^{K_g^d} \alpha_{k_g^d}^2  tr\left ( \mathbf{V}_{k_g^d} \mathbf{V}_{k_g^d}^H \right ) \leq P_{bs}, \; \forall \; g  \\
                        \label{eq:sopt_c2}
                        &  \; \; \; \; \;\;  (C.2) \;\; \gamma_{k_g^u}^2  tr\left ( \mathbf{V}_{k_g^u} \mathbf{V}_{k_g^u}^H \right ) \leq P_{ue}, \; \forall \; k,g.
\end{align}
Similarly, we augment the objective function with a weighted sum of the constraint functions, yielding the Lagrange function
\begin{equation}
\label{eq:opt_lagrange_s2}
    \begin{split}
    \mathcal{L} & \left ( \alpha, \gamma, \lambda \right ) \\
    & = \sum_{g=1}^{G}\sum_{k=1}^{K_g^d} \varepsilon_{k_g^d}  \left ( \alpha, \gamma, \mathbf{V}, \mathbf{U} \right ) + \sum_{g=1}^{G}\sum_{k=1}^{K_g^u} \varepsilon_{k_g^u} \left ( \alpha, \gamma, \mathbf{V}, \mathbf{U} \right ) \\
   & \quad + \sum_{g=1}^{G}  \lambda_{g} \left ( \sum_{k=1}^{K_g^d} \alpha_{k_g^d}^2 tr \left ( \mathbf{V}_{k_g^d} \mathbf{V}_{k_g^d}^H \right ) -P_{bs} \right ) \\
   &  \quad + \sum_{g=1}^{G} \sum_{k=1}^{K_g^u} \lambda_{k_g^u} \left ( \gamma_{k_g^u}^2 tr \left ( \mathbf{V}_{k_g^u} \mathbf{V}_{k_g^u}^H \right ) -P_{ue} \right ), 
    \end{split}
\end{equation}
where $\lambda_g \geq 0$ and $ \lambda_{k_g^u} \geq 0$ are Lagrange multipliers associated with constraints \eqref{eq:sopt_c1}-\eqref{eq:sopt_c2}. The Lagrange function is convex to the power coefficients. Differentiate the Lagrange function \eqref{eq:opt_lagrange_s2} with respect to $\alpha_{k_g^d}$ and $\gamma_{k_g^u}$ and set the derivatives to zero, the optimal power coefficients for the sub-problem $(S.2)$ are given as
\begin{align}
\alpha_{k_{g}^{d}}^{*} & =\frac{\mathcal{R}\left\{ tr\left(\mathbf{U}_{k_{g}^{d}}^{H}\hat{\mathbf{H}}_{k_{g}^{d},g}\mathbf{V}_{k_{g}^{d}}\right)\right\} }{\chi_{k_{g}^{d}}+\lambda_{g}^{*}  tr\left ( \mathbf{V}_{k_g^d} \mathbf{V}_{k_g^d}^H \right ) }\label{eq:slt_powCoe_dl},\\
\gamma_{k_{g}^{u}}^{*} & =\frac{\mathcal{R}\left\{ tr\left(\mathbf{U}_{k_{g}^{u}}^{H}\hat{\mathbf{H}}_{g,k_{g}^{u}}\mathbf{V}_{k_{g}^{u}}\right)\right\} }{\chi_{k_{g}^{u}}+\lambda_{k_{g}^{u}}^{*}  tr\left ( \mathbf{V}_{k_g^u} \mathbf{V}_{k_g^u}^H \right )}\label{eq:slt_powCoe_ul},
\end{align}
where $\chi_{k_g^d}$ and $\chi_{k_g^u}$ can be denoted using the function $\mathcal{F}_{2}\left( \cdot \right)$ defined in Appendix \ref{appendix:functions} as \\
\begin{equation}
    \begin{split}
    \chi_{k_{g}^{d}} & = \sum_{j=1}^{G}\sum_{i=1}^{K_{j}^{d}}\mathcal{F}_{2}\left(\mathbf{U}_{i_{j}^{d}},\hat{\mathbf{H}}_{i_{j}^{d},g},\mathbf{V}_{k_{g}^{d}}\right) \\
    & \quad +\sum_{j=1}^{G}\sum_{i=1}^{K_{j}^{u}}\mathcal{F}_{2}\left(\mathbf{U}_{i_{j}^{u}},\hat{\mathbf{H}}_{j,g},\mathbf{V}_{k_{g}^{d}}\right),
    \end{split}
\end{equation}
\begin{equation}
    \begin{split}
    \chi_{k_{g}^{u}} & = \sum_{j=1}^{G}\sum_{i=1}^{K_{j}^{d}}\mathcal{F}_{2}\left(\mathbf{U}_{i_{j}^{d}},\hat{\mathbf{H}}_{i_{j}^{d},k_{g}^{u}},\mathbf{V}_{k_{g}^{u}}\right) \\
    & \quad +\sum_{j=1}^{G}\sum_{i=1}^{K_{j}^{u}}\mathcal{F}_{2}\left(\mathbf{U}_{i_{j}^{u}},\hat{\mathbf{H}}_{j,k_{g}^{u}},\mathbf{V}_{k_{g}^{u}}\right).
    \end{split}
\end{equation}
To find the Lagrange multiplier $\lambda_{k_g^u}$, we can differentiate the Lagrange function \eqref{eq:opt_lagrange_s2} with respect to $\lambda_{k_g^u}$ and set the derivative to zero as
\begin{align*}
    &\gamma_{k_g^u}^2 tr \left ( \mathbf{V}_{k_{g}^{u}} \mathbf{V}_{k_{g}^{u}}^H \right ) -P_{ue} \\
    & = \left(\frac{\mathcal{R}\left\{ tr\left(\mathbf{U}_{k_{g}^{u}}^{H}\hat{\mathbf{H}}_{g,k_{g}^{u}}\mathbf{V}_{k_{g}^{u}}\right)\right\} }{\chi_{k_{g}^{u}}+\lambda_{k_{g}^{u}}^{*} tr \left ( \mathbf{V}_{k_{g}^{u}} \mathbf{V}_{k_{g}^{u}}^H \right )}\right)^{2} tr \left ( \mathbf{V}_{k_{g}^{u}} \mathbf{V}_{k_{g}^{u}}^H \right ) -P_{ue} = 0,
\end{align*}
yielding the optimal Lagrange multiplier as 
\begin{align*}
   &\;\; \lambda_{k_{g}^{u}}^{*} \\
   & = \max\left[0,-\frac{\chi_{k_{g}^{u}}}{ tr \left ( \mathbf{V}_{k_{g}^{u}} \mathbf{V}_{k_{g}^{u}}^H \right )}+\frac{\mathcal{R}\left\{ tr\left(\mathbf{U}_{k_{g}^{u}}^{H}\hat{\mathbf{H}}_{g,k_{g}^{u}}\mathbf{V}_{k_{g}^{u}}\right)\right\} }{\sqrt{ tr \left ( \mathbf{V}_{k_{g}^{u}} \mathbf{V}_{k_{g}^{u}}^H \right ) P_{ue}}}\right] ,
\end{align*}
where $\max\left[0, x \right]$ guarantees $ \lambda_{k_{g}^{u}} \geq 0$ to strictly satisfy the constraint \eqref{eq:sopt_c2}. For $\lambda_{g}$, it can be obtained from the positive root of $\sum_{k=1}^{K_{g}^{d}}\left(\frac{\mathcal{R}\left\{ tr\left(\mathbf{U}_{k_{g}^{d}}^{H}\hat{\mathbf{H}}_{k_{g}^{d},g}\mathbf{V}_{k_{g}^{d}}\right)\right\} }{\chi_{k_{g}^{d}}+\lambda_{g}tr \left ( \mathbf{V}_{k_{g}^{d}} \mathbf{V}_{k_{g}^{d}}^H \right )}\right)^{2}-\frac{P_{bs}}{tr \left ( \mathbf{V}_{k_{g}^{d}} \mathbf{V}_{k_{g}^{d}}^H \right )} = 0$ according to the derivative, or $\lambda_{g}^{*} = 0$ if the positive root does not exist. However, it is difficult to derive the close-form expression of its root since the variable is in the denominator. Alternatively, $\lambda_{g}^{*}$ can be obtained using Bisection search within the search range of $\left[0,\bar{\lambda}_{g}\right]$, where $\bar{\lambda}_{g}$ is the upper bound of $\lambda_{g}$. The power constraints at the BS can be written as Equation \eqref{eq:powCst_1}, which gives Equation \eqref{eq:powCst_2}. Equation \eqref{eq:powCst_2} is in the form of $a\lambda_g^2 + 2b\lambda_g + c \leq 0$, where $a = \sum_{k=1}^{K_{g}^{d}} tr\left(\mathbf{V}_{k_{g}^{d}}\mathbf{V}_{k_{g}^{d}}^H\right)$, $b = \sum_{k=1}^{K_{g}^{d}} \chi_{k_{g}^{d}}$, and $c = \sum_{k=1}^{K_{g}^{d}}\frac{\chi_{k_{g}^{d}}^{2}}{tr\left(\mathbf{V}_{k_{g}^{d}}\mathbf{V}_{k_{g}^{d}}^H\right)} - \frac{1}{P_{bs}}\sum_{k=1}^{K_{g}^{d}}\mathcal{R}\left\{ tr\left(\mathbf{U}_{k_{g}^{d}}^{H}\hat{\mathbf{H}}_{k_{g}^{d},g}\mathbf{V}_{k_{g}^{d}}\right)\right\} ^{2}$. Thus, we can derive the upper bound as
\begin{align*}
    \lambda_{g}\leq  -\frac{b}{a} + \sqrt{\frac{b^2}{a^2} - \frac{c}{a}} =  \bar{\lambda}_g .
\end{align*}
\begin{figure*}
    \normalsize
    \setcounter{mytempeqncnt}{\value{equation}}
    \begin{equation}
     \label{eq:powCst_1}
      \begin{split}
       0 = P_{bs}-\sum_{k=1}^{K_{g}^{d}} & \left( \frac{\mathcal{R}\left\{ tr\left(\mathbf{U}_{k_{g}^{d}}^{H}\hat{\mathbf{H}}_{k_{g}^{d},g}\mathbf{V}_{k_{g}^{d}}\right)\right\} }{\chi_{k_{g}^{d}} +\lambda_{g}tr\left(\mathbf{V}_{k_{g}^{d}}\mathbf{V}_{k_{g}^{d}}^H\right)}\right)^{2} tr\left(\mathbf{V}_{k_{g}^{d}}\mathbf{V}_{k_{g}^{d}}^H\right)  \leq P_{bs}-\frac{\sum_{k=1}^{K_{g}^{d}}\mathcal{R}\left\{ tr\left(\mathbf{U}_{k_{g}^{d}}^{H}\hat{\mathbf{H}}_{k_{g}^{d},g}\mathbf{V}_{k_{g}^{d}}\right)\right\} ^{2} tr\left(\mathbf{V}_{k_{g}^{d}}\mathbf{V}_{k_{g}^{d}}^H\right)}{\sum_{k=1}^{K_{g}^{d}}\left(\chi_{k_{g}^{d}}+\lambda_{g}tr\left(\mathbf{V}_{k_{g}^{d}}\mathbf{V}_{k_{g}^{d}}^H\right)\right)^{2}} \\
       & =P_{bs} -\frac{\sum_{k=1}^{K_{g}^{d}}\mathcal{R}\left\{ tr\left(\mathbf{U}_{k_{g}^{d}}^{H}\hat{\mathbf{H}}_{k_{g}^{d},g}\mathbf{V}_{k_{g}^{d}}\right)\right\} ^{2} tr\left(\mathbf{V}_{k_{g}^{d}}\mathbf{V}_{k_{g}^{d}}^H\right)}{\lambda_{g}^{2} \sum_{k=1}^{K_{g}^{d}} \left ( tr\left(\mathbf{V}_{k_{g}^{d}}\mathbf{V}_{k_{g}^{d}}^H\right) \right )^{2}+2\lambda_{g} \sum_{k=1}^{K_{g}^{d}} tr\left(\mathbf{V}_{k_{g}^{d}}\mathbf{V}_{k_{g}^{d}}^H\right)\chi_{k_{g}^{d}}+\sum_{k=1}^{K_{g}^{d}}\chi_{k_{g}^{d}}^{2}}
         \end{split}
  \end{equation}
    \hrulefill
    \vspace*{4pt}
\end{figure*}
\begin{figure*}
    \normalsize
    \setcounter{mytempeqncnt}{\value{equation}}
    \begin{equation}
     \label{eq:powCst_2}
      \begin{split}
       \lambda_{g}^{2} \sum_{k=1}^{K_{g}^{d}} tr\left(\mathbf{V}_{k_{g}^{d}}\mathbf{V}_{k_{g}^{d}}^H\right)+2\lambda_{g} \sum_{k=1}^{K_{g}^{d}} \chi_{k_{g}^{d}}+\sum_{k=1}^{K_{g}^{d}}\frac{\chi_{k_{g}^{d}}^{2}}{tr\left(\mathbf{V}_{k_{g}^{d}}\mathbf{V}_{k_{g}^{d}}^H\right)} - \frac{1}{P_{bs}}\sum_{k=1}^{K_{g}^{d}}\mathcal{R}\left\{ tr\left(\mathbf{U}_{k_{g}^{d}}^{H}\hat{\mathbf{H}}_{k_{g}^{d},g}\mathbf{V}_{k_{g}^{d}}\right)\right\} ^{2} \leq 0
         \end{split}
  \end{equation}
    \hrulefill
    \vspace*{4pt}
\end{figure*}

\vspace{-3mm}
% -------------------------------------------------
\subsection{Iterative Processing}
\label{sec:solutions}
The solutions of beamforming matrices and power coefficients in Equations \eqref{eq:slt_combiner_dl}, \eqref{eq:slt_combiner_ul}, \eqref{eq:slt_precoder_dl}, \eqref{eq:slt_precoder_ul}, \eqref{eq:slt_powCoe_dl}, and \eqref{eq:slt_powCoe_ul} show inter-dependence on each other. Thus, we need an iterative algorithm to obtain the overall solution by updating them until they converge. We first initialize power coefficients to meet the power budgets of each node and randomly initialize beamforming matrices. Then, we continuously update the precoding matrices, power coefficients, and combining matrices in order. The iterative procedure stops if the loss function, i.e., the objective function in \eqref{eq:opt_problem}, does not decrease with iterations (the decreasing amount is less than the threshold) or the iteration time exceeds the limitation. This processing is summarized as Algorithm \ref{alg:iterative_JPAIM}. Note that, since the joint optimization problem is not jointly convex to the power coefficients and beamforming matrices, the algorithm can only converge to the local minimal closest to the initial point. Thus, the performance strongly depends on the initial point, which is similar to existing studies \cite{paulaWsr, jointMu}.
\begin{algorithm}[t]
	\caption{Iterative JPAIM Algorithm}
	\label{alg:iterative_JPAIM}
	\begin{algorithmic}[1]
		\STATE Initialize the power allocation coefficients as $\alpha_{k_g^d} = \sqrt{\frac{P_{bs}}{b_d K_g^d}}$ and $\gamma_{k_g^u} = \sqrt{\frac{P_{ue}}{b_u}}$ $\forall \, k,g$.
		\STATE Randomly initializes the precoding matrices and normalize.
		\REPEAT 
		\STATE Update the combining matrices $\left\{ \mathbf{U}_{k_g^d}, \mathbf{U}_{k_g^u} \right\}_{\forall k,g}$ as \eqref{eq:slt_combiner_dl} and \eqref{eq:slt_combiner_ul}.
		\STATE Calculate and record the current sum of MSE as $\varepsilon^{(t)}$
		\STATE Update the precoding matrices $\left\{ \mathbf{V}_{k_g^d}, \mathbf{V}_{k_g^u} \right\}_{\forall k,g}$ as \eqref{eq:slt_precoder_dl} and \eqref{eq:slt_precoder_ul}.
		\STATE Update the power allocation coefficients $\left\{ \alpha_{k_g^d}, \gamma_{k_g^u} \right\}_{\forall k,g}$ as \eqref{eq:slt_powCoe_dl} and \eqref{eq:slt_powCoe_ul}.
		\UNTIL {$\varepsilon^{(t)} - \varepsilon^{(t-1)} <  \iota$ or $t>t_{m}$, where $\iota$ is the threshold and $t_{m}$ is the iteration limit.}
	\end{algorithmic}  
\end{algorithm}

\subsubsection{Convergence Behaviour}
This algorithm can keep reducing the objective function of the joint optimization problem (i.e., loss function) and converge to a local minimum. Let $\theta \left( \alpha,\beta,\mathbf{V},\mathbf{U} \right)$ denote the loss function. Within an iteration, one of the optimized variables is updated with others fixed, and the Lagrange function is convex to the target variable under this condition. Thus, we always obtain the optimal value of the target variable with other variables fixed, reducing the loss function at each step as:
\begin{align}
    \theta \left( \alpha,\gamma,\mathbf{V},\mathbf{U}^{(t+1)} \right) & \leq \theta \left( \alpha,\gamma,\mathbf{V},\mathbf{U}^{(t)} \right) \label{eq:conv_1}\\
    \theta \left( \alpha,\gamma,\mathbf{V}^{(t+1)},\mathbf{U} \right) & \leq \theta \left( \alpha,\gamma,\mathbf{V}^{(t)},\mathbf{U} \right) \label{eq:conv_2}  \\
    \theta \left( \alpha^{(t+1)},\gamma^{(t+1)},\mathbf{V},\mathbf{U} \right) & \leq \theta \left( \alpha^{(t)},\gamma^{(t)},\mathbf{V},\mathbf{U} \right)   \label{eq:conv_3}
\end{align}
After the $t^{th}$ iteration, we have $ \alpha^{(t)},\gamma^{(t)},\mathbf{V}^{(t)},\mathbf{U}^{(t)} $, and then update the variables to $\mathbf{U}^{(t+1)}, \mathbf{V}^{(t+1)}, \alpha^{(t+1)},\gamma^{(t+1)}$ in order. According to \eqref{eq:conv_1}-\eqref{eq:conv_3}, the loss function is guaranteed to be reduced at each iteration until it converges to a local optimal.

\subsubsection{Computational Complexity}
\label{subsec:complexity}
We measure the computational complexity by accounting for the required multiplication operations. The required multiplications of some basic operations are given in Table \ref{table:complexity1}, where $\mathbf{A}_1$ is of size $a \times b$ and $\mathbf{A}_2$ is of size $b \times c$.
\begin{table}[h]
\centering
\caption{Complexity of Operations\label{table:complexity1}}
\begin{tabular}{|c|c|}
\hline
\multicolumn{1}{|c|}{Operation}                                                       & \multicolumn{1}{|c|}{Number of multiplications}                                                                                     \\ \hline
$\mathbf{A}_1 \mathbf{A}_2$   &   $a \times b \times c$     \\ \hline
Eigen decomposition $\left ( \mathbf{A}_1 \right )$   & $\mathcal{O}\left (a^3 \right )$    \\ \hline
Matrix inverse $\left ( \mathbf{A}_1 \right )$   &  $\mathcal{O}\left ( a^3 \right )$    \\ \hline
$\mathcal{F}_1\left ( \mathbf{A}_1, \mathbf{A}_2 \right )$   &   $2a^2b+ab^2+b^2c+2abc+3a^2$     \\ \hline
\end{tabular}
\end{table}
According to this, the total number of multiplications of calculating beamforming matrices (taking $\mathbf{V}_{k_g^d}$ as an instance) can be given as
\begin{align*}
    M_v & = \underbrace{GK \large [ 3A_b^3 + A_b^2(2A_u+3N_s+6) + A_b(A_u^2 +2A_uN_s)}_{\mathbf{\Omega}_g} \\
    &\quad \underbrace{+ A_u^2N_s \large ]}_{\mathbf{\Omega}_g} + \underbrace{A_b^3}_{\mathbf{H}_{g,g}^{H} \mathbf{H}_{g,g}} + \underbrace{A_b^2A_u+A_bA_uN_s}_{\text{rest multiplication}} ,
\end{align*}
plus an Eigen decomposition and an inverse to a matrix of size $A_b \times A_b$, where the calculation process is divided into calculating $\mathbf{\Omega}_g$, $mathbf{H}_{g,g}^{H} \mathbf{H}_{g,g}$, and rest matrix multiplication, and we assume $K_g^d=K_g^u=K,\forall \; g$, $N_{bs}=M_{bs}=A_b$, $N_{ue}=M_{ue}=A_u$, and $b_d=b_u=N_s$ for simplicity. Usually, we will have $A_b\gg A_u \geq N_s$. To update the corresponding power coefficient $\alpha_{k_g^d}$, the complexity mainly comes from calculating $\chi_{k_g^d}$ and the rest matrix multiplication. Thus, the total number of multiplication operations required to compute $\mathbf{\alpha}_{k_g^d}$ is
\begin{align*}
    M_{\alpha} &= \underbrace{GK[ 2A_b^3 + A_b^2(A_u+5N_s+2)+A_b(A_u^2 + 4A_uN_s}_{\chi_{k_g^d}} \\
    & \quad \underbrace{+N_s^2)+A_u^2N_s+2A_uN_s+2N_s^2]}_{\chi_{k_g^d}} \\
    & \quad + \underbrace{A_b^2N_s + A_b(A_uN_s+N_s^2)}_{\text{rest multiplication}} .
\end{align*}
Therefore, the number of multiplication operations required for one iteration of the JPAIM algorithm is given as $M_v+M_{\alpha}$ plus an Eigen decomposition and inverse operation to a matrix of size $A_b \times A_b$, which is in the order of $GKA_b^3$ such that
\begin{align}
    C_{\text{JPAIM}} & = \mathcal{O} \left ( GKA_b^3  \right ) .
\end{align}
The computational complexity is in the same order as the MWSR algorithm proposed in \cite{paulaWsr}. The complexity analysis of the MWSR algorithm is not detailed in this paper since it is out of the scope, but it can be analyzed similarly based on the solution given in \cite{paulaWsr}. Although the complexities of the two algorithms are in the same order, JPAIM needs $M_{\alpha}+A_b^3$ more multiplications, which almost double the time complexity of a single iteration. Therefore, the JPAIM algorithm needs at least half the convergence time to reduce the overall complexity.

\begin{table*}[]
\centering
\caption{Simulation Setting\label{table:sim_setting}}
\begin{tabular}{|l|c|}
\hline
\multicolumn{1}{|c|}{Parameters}                                                       & \multicolumn{1}{|c|}{Values}                                                                                     \\ \hline
Number of antennas    &  $N_{bs} = M_{bs} = 16$, $N_{ue} = M_{ue} = 2$ if not specified     \\ \hline
Cells   &    hexagon cells with 200m of inter-cite distance and 10m of a minimum base station to users distance  \\ \hline
LOS probability and pathloss   &  models for UMi scenarios in \cite{3gppChannel}     \\ \hline
Channel matrices   &  \begin{tabular}[c]{@{}c@{}} $\mathbf{H}_{B,A}= \sqrt{l_{B,A}} \, \mathbf{N}$ if the LOS probability $\geq0.5$ \\ $\mathbf{H}_{B,A}= \sqrt{l_{B,A}} \left ( \sqrt{\frac{k_{B,A}}{k_{B,A}+1}} \, \mathbf{I} + \sqrt{\frac{1}{k_{B,A}+1}} \, \mathbf{N} \right )$ if the LOS probability $<0.5$  \\ ($\mathbf{N}$ has elements independently and identically distributed as $\mathcal{CN}(0,1)$)  \end{tabular}    \\ \hline
Transmit power budgets   &    24dBm for BSs and 23dBm for UEs   \\ \hline
Thermal noise density   &     -174dBm/Hz   \\ \hline
Noise figure  &  13dB for BSs and 9dB for UEs     \\ \hline
Number of effective bits of DACs/ADCs  &   $b=12$ if not specified    \\ \hline
Channel uncertainty factor   &  $\hat{\sigma}_{B,A}^2 = \varrho \left| \mathbf{H}_{B,A} \right|_F^2$ $\forall A,B$, where $\varrho = -120$dB     \\ \hline
\end{tabular}
\end{table*}
% ----------------------------------------------------------------------------------------------------------------------------------------------------------
\section{Simulation Results}
\label{sec:results}
We consider OFDM modulation with 15kHz subcarrier spacing and 10 MHz bandwidth centered at 2.5GHz, and other parameters are given in Table \ref{table:sim_setting}. The system parameters are chosen based on 3GPP specifications \cite{3gppTxrx, 3gppChannel} and our preliminary work \cite{paulaWsr, luo2023beam}. The ASIC depth realized by any other methods (e.g., antenna isolation as in \cite{antennaSic} or RF cancellers as in \cite{hlAccess}) will be directly reflected by the pathloss of the SI channel regardless of the implementation methods. We assume the pathloss without any ASIC is 0dB due to the proximity of IBFD transceivers \cite{surveySic}. Therefore, $x$ dB of realized ASIC by other techniques means $l_g = x$ dB $\forall g$ in simulations, where $l_g$ is defined in Section \ref{subsec:bf_asic}. We do not presume a realized ASIC depth but evaluate the algorithms' performance at different ASIC depths. In the later simulations, we consider the ideal case, i.e., perfect ASIC with $l_g=120$ dB, to demonstrate the characteristics of the algorithm itself, but this does not mean that our algorithm requires such perfect ASIC.

Since the throughput maximization for MCMU networks is not convex and the performance depends on the initial point, we run 1,000 times Monte Carlo simulations for each scenario and evaluate the performance via the averaged values. The realization is randomly generated at each step of the loop. For a fair comparison, the simulation settings, e.g., user location, channel matrices, initial points, etc., remain identical for the two algorithms in each realization in 1,000 iterations of the Monte Carlo simulations, and the two algorithms will be performed separately.

\begin{figure*}
      \centering
      \includegraphics[width=1.0\textwidth]{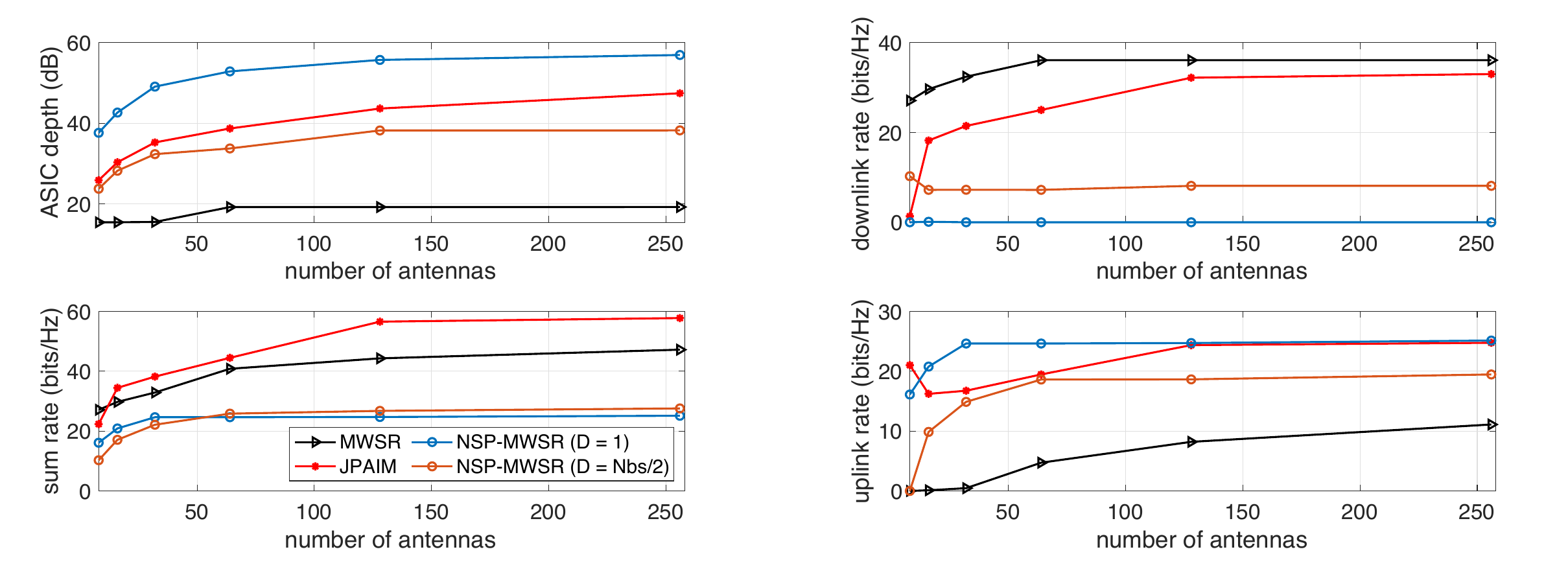}
      \caption{The ASIC capability and achievable spectral efficiency versus the number of transmitting and receiving antennas of the FD node.}
      \label{fig:sic_cap}
\end{figure*}
\subsection{SIC Capability}
Fig. \ref{fig:sic_cap} shows the achievable ASIC depth and SE versus the number of transmitting and receiving antennas of IBFD transceivers using different beamforming schemes. In order to demonstrate the ASIC capability of precoders, we consider the scenario that no ASIC is realized by other techniques since precoders will tend to suppress other significant interference rather than SI if it has been efficiently suppressed. In this case, the MWSR precoder optimizes the interference management for the downlink payload regardless of SI, resulting in a large received SI  power to saturate the IBFD receiver and significantly degrade the uplink rate. More receiving antennas compensate for the uplink rate reduction due to more degrees of freedom of combiners, managing the interference better for the uplink payload. Combining the null-space projection \cite{tbSic} with MWSR beamforming as explained in Appendix \ref{appendix:NSP-MMSE} can enhance its ASIC capability. The uplink rate is increased at the cost of the downlink rate degradation. The trade-off can be made by adjusting the dimensions of the subspace chosen for projection, i.e., $D$ (defined in Appendix \ref{appendix:NSP-MMSE}), which also determines the ASIC capability. The disadvantage of NSP-MWSR is that the downlink rate decreases significantly with increasing $D$. In contrast, our proposed JPAIM beamforming has less impact on the downlink rate with enhanced ASIC capability. It can achieve higher SE than other schemes and simultaneously support effective downlink and uplink communications with $\geq 16$ antennas at the IBFD transceiver. It can offer $>30$dB of ASIC depth with $\geq16$ transmitting antennas at the IBFD node and can be improved by more antennas. When the antenna array is small (e.g., 8-antenna arrays), the power allocation tends to reduce the transmit power of the downlink payload to maximize the sum rate by maximizing the uplink rate. With enlarging antenna arrays (e.g., 16-antenna arrays), it is able to provide effective ASIC to enable efficient IBFD communications. However, it introduces interference for the uplink payload so that the uplink rate goes down at first.

\vspace{-3mm}
\subsection{Convergence Behaviour}
Fig. \ref{fig:conv} shows the convergence behavior of the JPAIM algorithm compared to the MWSR algorithm proposed in \cite{paulaWsr}. Since the joint power allocation and beamforming optimization problem is not jointly convex, the convergence behavior and achievable sum rate depend on the initial points as the algorithms converge to the local minimum closest to the initial point. Therefore, we ran 1,000 times Monte Carlo simulations and calculated the average value to show their performance. It shows that the JPAIM algorithm converges faster than the MWSR algorithm under both conditions (i.e., insufficient and sufficient ASIC realized). The JPAIM algorithm reduces the iteration time to converge by at least half compared to the MWSR Algorithm, which meets the condition analyzed in Section \ref{subsec:complexity} to reduce the overall time complexity. Besides, both algorithms converge faster with insufficient ASIC realized.
\begin{figure}
      \centering
      \includegraphics[width=0.5\textwidth]{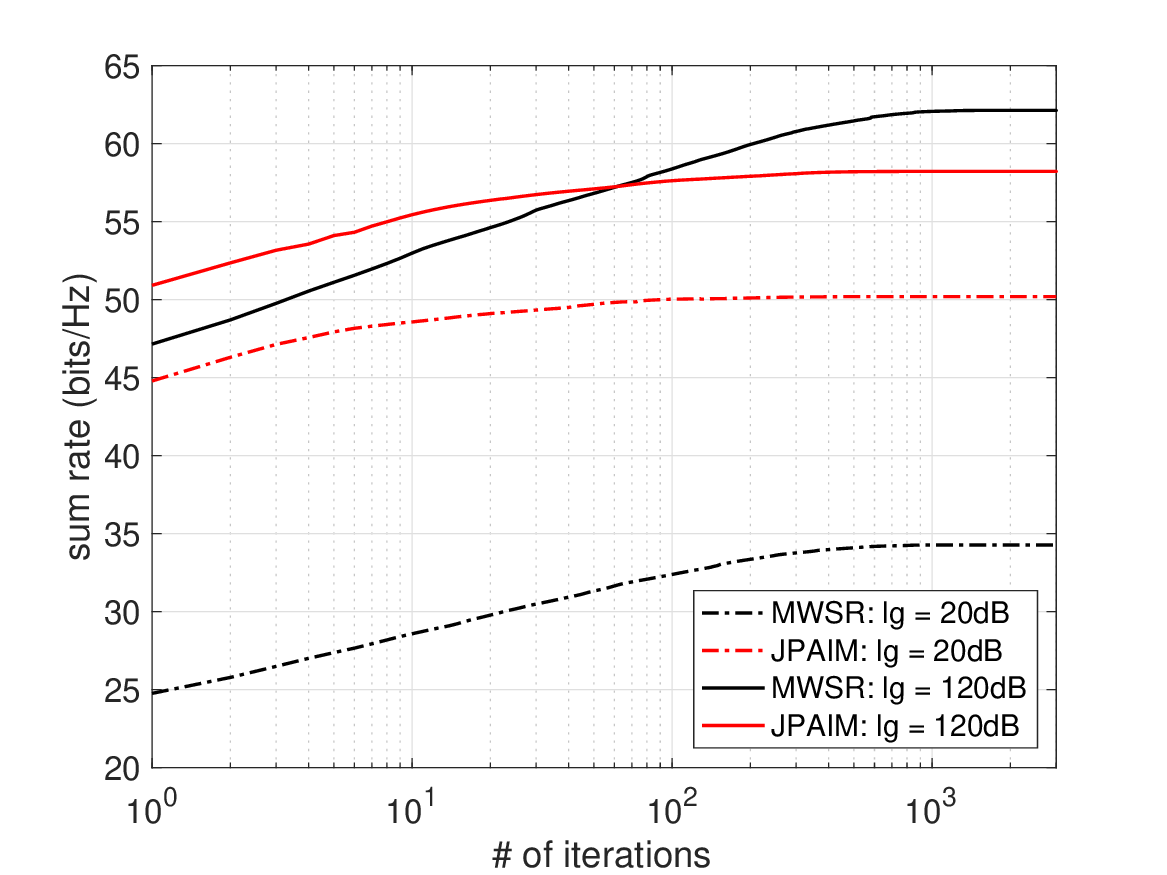}
      \caption{Convergence behavior (x-axis is in the log scale).\label{fig:conv}}
\end{figure}

\vspace{-3mm}
\subsection{Improvement with Insufficient ASIC}
In this section, we evaluate the performance of the JPAIM algorithm in terms of achievable SE (i.e., sum rate), as the most attractive factor of IBFD is its potential to double the SE of its HD counterpart. Fig. \ref{fig:se_com} shows the achievable SE of the JPAIM algorithm versus realized ASIC depths compared to the MWSR algorithm. When the ASIC depth realized by other techniques is insufficient (i.e., $<60$dB), the proposed JPAIM achieves higher SE than MWSR. This is attractive as other ASIC techniques have incredibly high complexity and cost in MIMO systems, as we stated in Section \ref{sec:intro}. Active antenna isolation, which has feasible implementation complexity, can achieve around 30 dB of ASIC in practice \cite{antennaSic}. In contrast, deeper ASIC needs to be realized by the much more complex RF cancellation, which is almost practically prohibited in MIMO systems. In this case, the JPAIM algorithm can improve the SE by $42.9\%$ in IBFD compared to HD, while the MWSR algorithm cannot. This value can be improved to $60.9\%$ with further digital processing if the received SI has already been within the dynamic range of receivers. In contrast, the MWSR must require RF cancellers to achieve the IBFD gain due to its limited ASIC capability, as illustrated in Fig. \ref{fig:sic_cap}.
\begin{figure}
      \centering
      \includegraphics[width=0.5\textwidth]{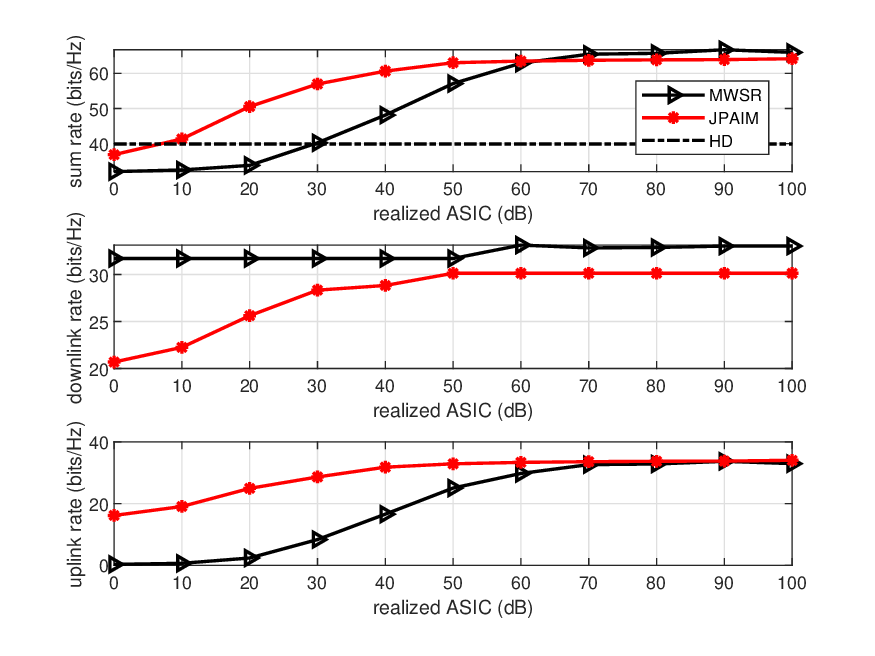}
      \caption{Achievable spectral efficiency vs realized ASIC depth (i.e., $l_g$).}
      \label{fig:se_com}
\end{figure}

\vspace{-3mm}
\subsection{Time Efficiency}
\begin{figure*}
    \centering
    \subfigure[Sum rate and computation time versus the number of users.]{
    \begin{minipage}[t]{0.5\linewidth}
    \centering
    \includegraphics[width=1.0\linewidth]{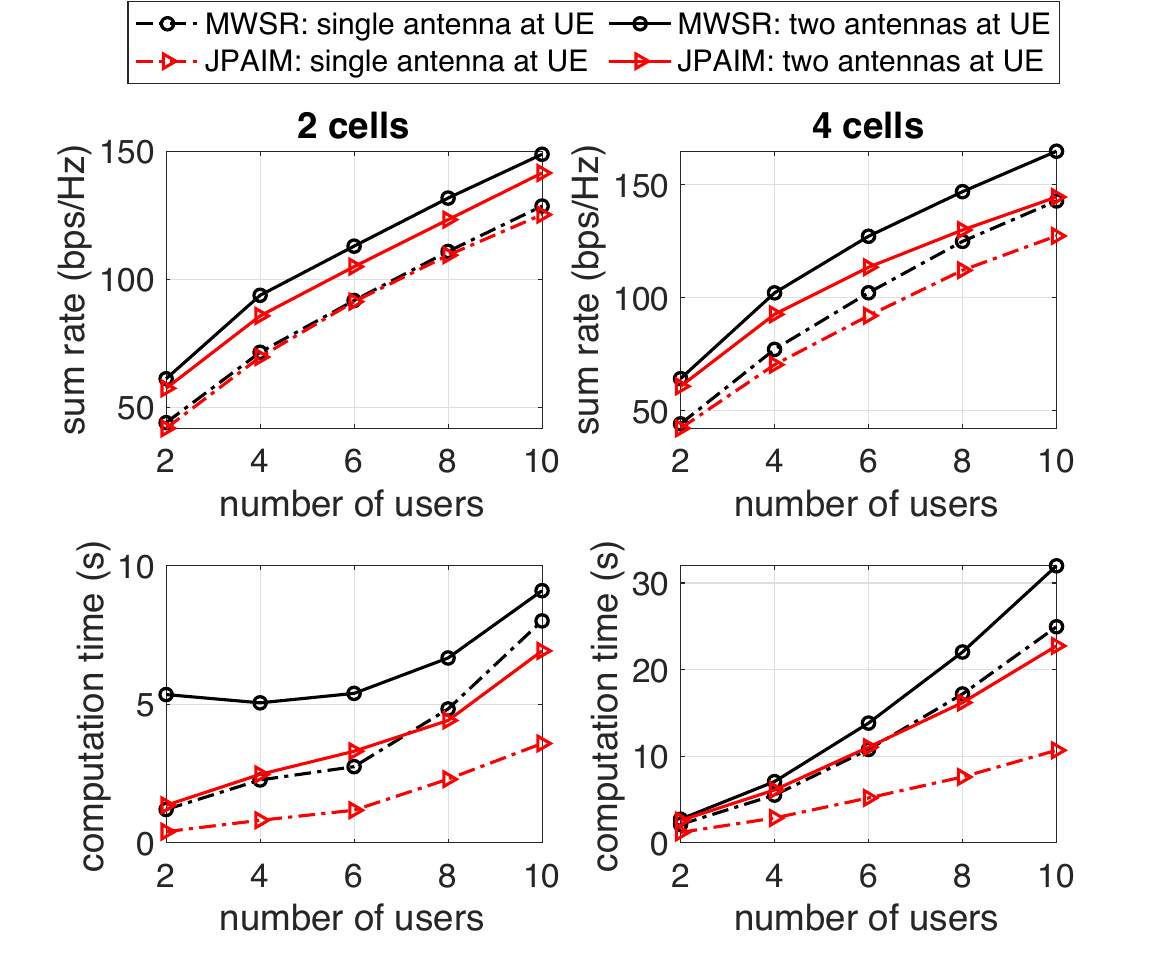}
    \label{fig:cc_1}
    \end{minipage}%
    }%
    \subfigure[Sum rate reduction and computation time save percentages.]{
    \begin{minipage}[t]{0.5\linewidth}
    \centering
    \includegraphics[width=1.0\linewidth]{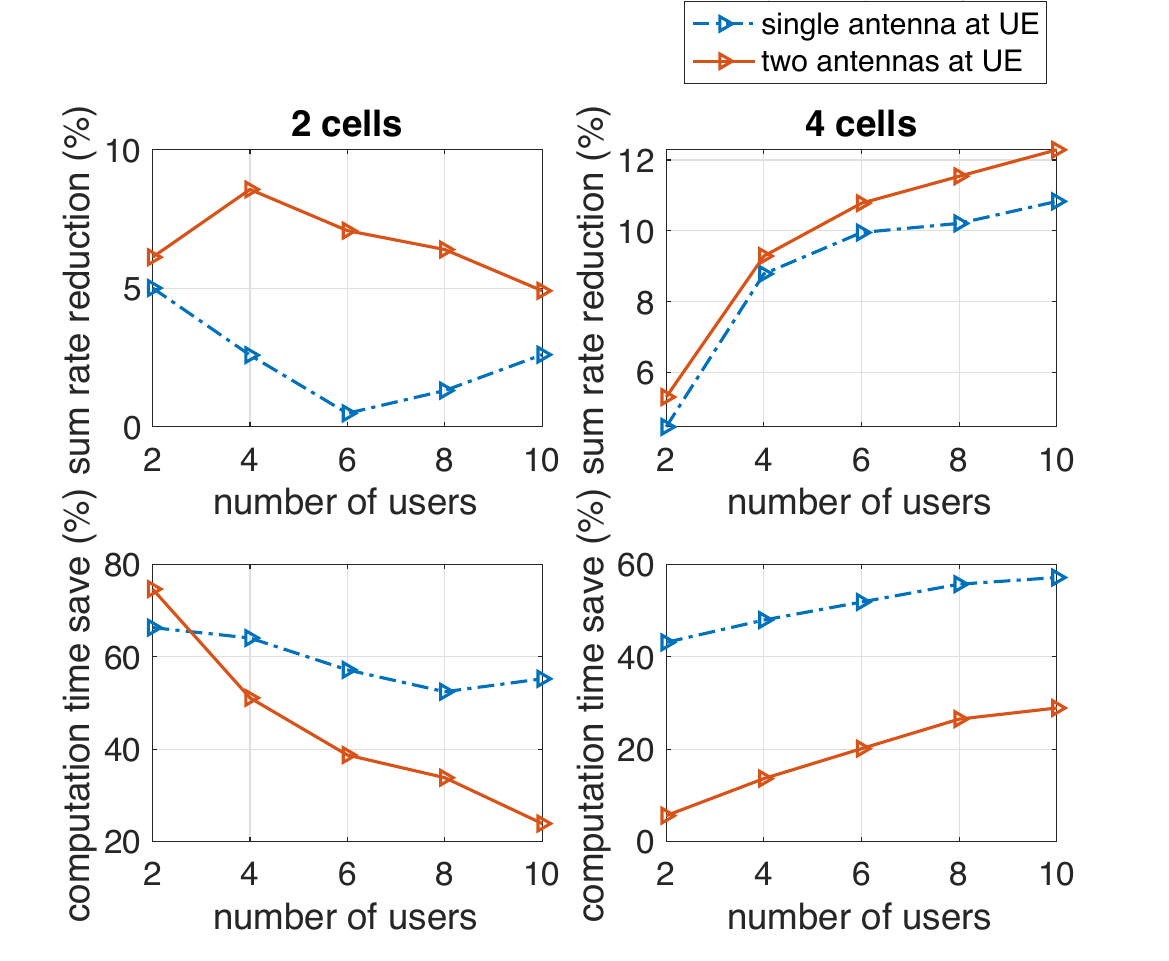}
    \label{fig:cc_2}
    \end{minipage}%
    }%
    \centering
    \caption{Performance comparison of JPAIM and MWSR in terms of SE and time complexity. \label{fig:cc}}
\end{figure*}
\begin{figure}
      \centering
      \includegraphics[width=0.5\textwidth]{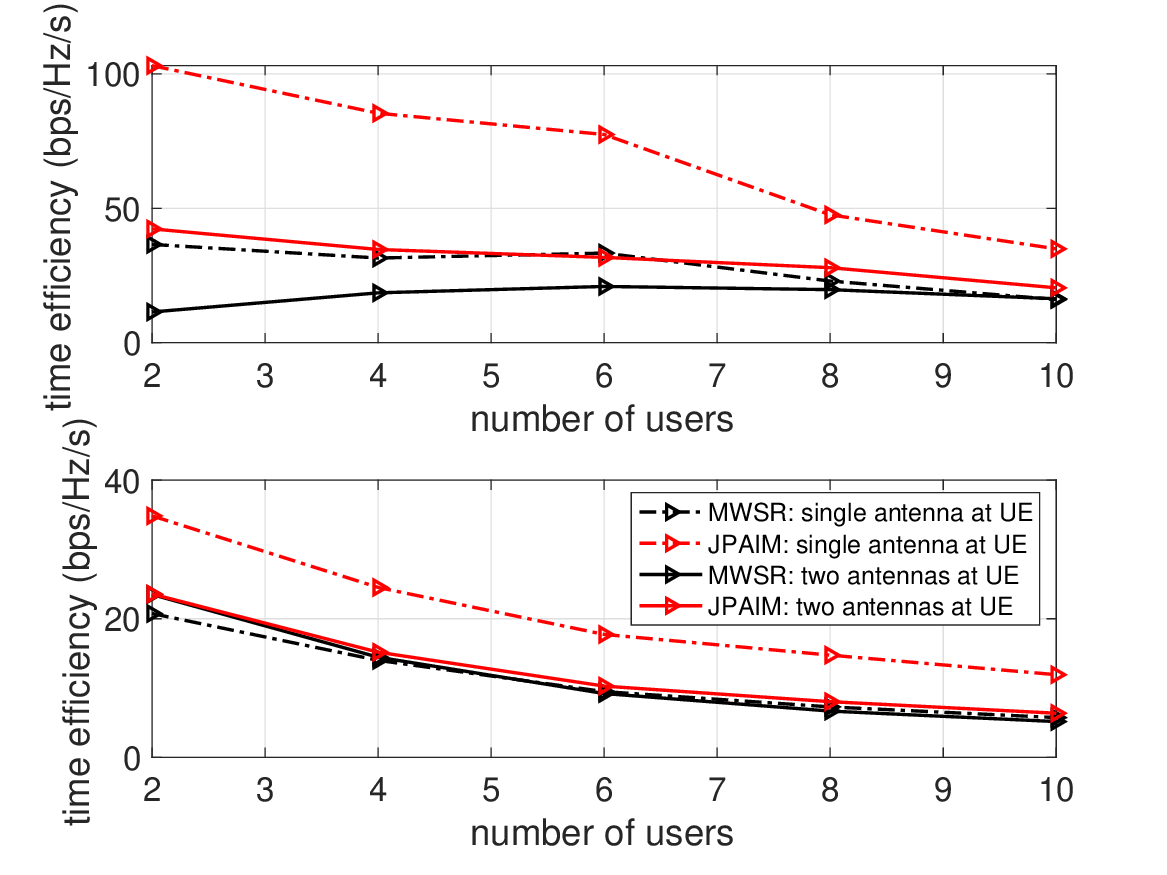}
      \caption{Time efficiency of the proposed JPAIM algorithm.\label{fig:efficiency}}
\end{figure}
Results in Fig. \ref{fig:se_com} suggest that MWSR can achieve higher SE than JPAIM with sufficient ASIC applied. This is reasonable since MWSR yields a joint power allocation and beamforming solution while JPAIM decomposes it into two sub-problems to solve. However, the decomposition also brings benefits in terms of time efficiency (i.e., $=\frac{\textup{sum rate}}{\textup{computation time}}$). In order to evaluate the performance of the proposed algorithm in terms of time efficiency, we compare the SE and time complexity of the two algorithms in networks of different scales. We use Monte Carlo simulations to compare 2-cell and 4-cell networks with 2-10 users.

Fig. \ref{fig:cc_1} shows the performance difference between the two algorithms in terms of sum rate and average computation time of a realization with different network sizes, i.e., different numbers of cells, users, and antennas at users. Fig. \ref{fig:cc_2} shows the corresponding sum rate reduction and computation save percentages of JPAIM compared to MWSR, and Fig. \ref{fig:efficiency} shows their time efficiency. The results show that JPAIM achieves a sum rate close to the one achieved by MWSR, but it takes much less computation time. The sum rate loss generally increases with increasing numbers of cells and users, and it is more significant with multi-antenna users. For single-antenna users, JPAIM saves at least $40\%$ of the computation time at the cost of $<10\%$ sum rate reduction for most cases. Although the benefits are compromised with two-antenna users, the computation time saved is always higher than the sum rate reduction. 

Fig. \ref{fig:efficiency} also illustrates that JPAIM achieves higher time efficiency than MWSR under all conditions. This indicates that JPAIM has a lower time complexity in practical implementations than MWSR at the cost of an acceptable sum rate loss. The benefit of JPAIM is significant with single-antenna users, which is meaningful for practical cellular network deployment since many UE devices still use a single antenna. This can be understood since our JPAIM algorithm uses a coarse power allocation scheme, while MWSR uses a refined power allocation. Therefore, MWSR can achieve higher SE than JPAIM. However, refined power allocation will drastically increase the computational complexity, making it less time-efficient than our proposed algorithm. When the interference is not complicated (i.e., with fewer cells and users), the refined power allocation is not so important, while a coarse power allocation can greatly reduce the time complexity due to its fast convergence characteristics, so JPAIM has a much higher time efficiency than MWSR. When the interference becomes complicated (i.e., in enlarging networks), a refined power allocation will significantly improve network capacity, while a coarse power allocation can reduce the complexity relatively less (as shown in Fig. \ref{fig:cc}). In addition, our algorithm demonstrates a significant time efficiency improvement for single-antenna users, while the improvement is less pronounced for multi-antenna users. The reason is that multi-antenna users can perform beamforming to increase channel capacity, which requires user terminals to also perform iterative calculations, thereby significantly increasing the computational complexity. This is not necessary for the JPAIM algorithm based on coarse power allocation when single antenna users are considered, while the refined power allocation-based MWSR algorithm still requires iterative calculations for power allocation of the data streams. Therefore, the time complexity benefit of JPAIM compared to MWSR is greatly reduced when there are multiple-antenna users (see Fig. \ref{fig:cc_2}), resulting in the overall time efficiency is not obvious.

\vspace{-3mm}
\subsection{Robustness to Channel Uncertainty and Transceiver HWIs}
\begin{figure}
      \centering
      \includegraphics[width=0.5\textwidth]{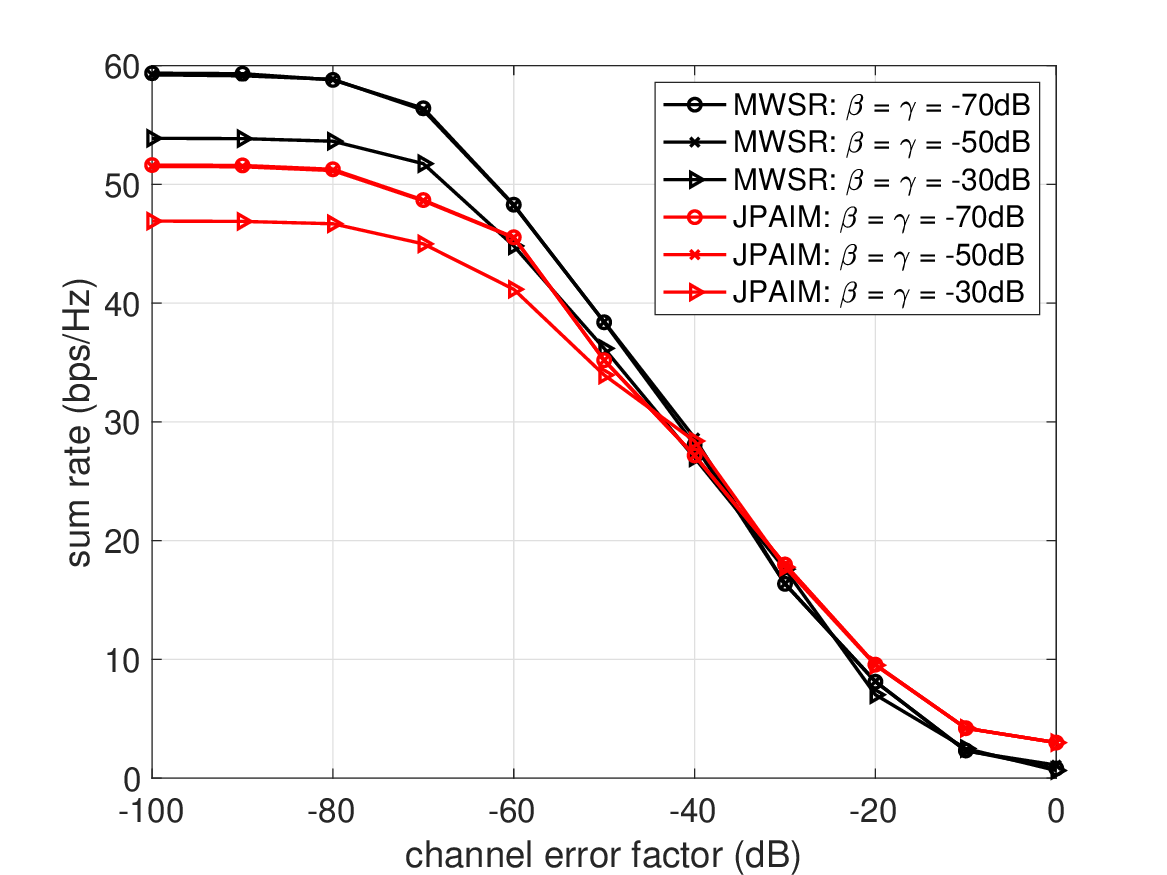}
      \caption{Robustness to channel uncertainty and transceiver HWIs.\label{fig:robustness}}
\end{figure}
In our study, the criteria for robustness are defined based on the algorithm's ability to maintain consistent performance levels under varying channel conditions and hardware impairments. That is, the algorithm can be considered robust if its performance does not degrade dramatically with increasing channel uncertainty and hardware impairments regardless of the highest sum rate. Fig. \ref{fig:robustness} shows the achievable sum rate variation of the JPAIM algorithm compared to the MWSR algorithm against the channel error factor $\varrho$. It demonstrates the robustness of JPAIM to channel uncertainty. It shows that significant channel uncertainty has less impact on JPAIM than MWSR. The benefits come from the fact that the effects of channel uncertainty are included in our algorithm as detailed in Appendix \ref{appendix:chError}. In addition, it is also robust to the transceiver hardware impairments that reducing the dynamic range of transceivers from 70dB to 50dB has no noticeable effect on the achievable sum rate. Although a smaller dynamic range (e.g., 30dB) of transceivers reduces the achievable sum rate, it is acceptable as the reduction is less than $10\%$. The robustness to transceiver HWIs is similar for MSER and JPAIM since they both consider the HWIs in their algorithm designs.

% ----------------------------------------------------------------------------------------------------------------------------------------------------------
\section{Conclusions}
\label{sec:conclusions}
In this paper, we have formulated and solved an MSE minimization problem regarding power allocation and beamforming with hardware limits. We have revealed that downlink precoders have to sufficiently suppress SI before the receiver (i.e., without combiners); otherwise, the uplink communication rate will be naturally reduced due to the limited dynamic range of receivers. To this end, we have enhanced the ASIC capability of precoders by adding a constraint on the received RSI power and have demonstrated the advantages of our method over existing NSP-based methods. The formulated problem has been converted to an MSE plus RSI minimization problem and then decomposed into two sub-problems to be solved. The closed-form solutions to the sub-problems are derived, and the overall solution is obtained by an iterative algorithm. It has been demonstrated that our JPAIM algorithm could achieve 42.9\% of the IBFD gain in terms of spectral efficiency in a 3GPP-specified multi-cell multi-user network with feasible antenna isolation only, inspiring a low-cost but efficient IBFD cellular network implementation. Then, we have evaluated its performance under ideal SIC conditions (i.e., sufficient ASIC depth) by comparing it to the existing MWSR algorithm \cite{paulaWsr}, showing it takes much less computation time due to its faster convergence speed at the cost of acceptable sum rate loss. The benefit of our proposed JPAIM algorithm is significant with single-antenna users that it saves at least $40\%$ of the computation time at the cost of $<10\%$ sum rate reduction. In addition, we have proved that it is robust to channel uncertainty and transceiver hardware impairments due to the robust design.

Future work could consider more practical scenarios, e.g., with limited CSI due to the limited backhaul resources.

% ----------------------------------------------------------------------------------------------------------------------------------------------------------
\section*{Acknowledgments}
The work was supported in part by the research grant from Huawei Technologies (Sweden) AB.

\appendices
% ----------------------------------------------------------------------------------------------------------------------------------------------------------
\section{Effects of Channel Uncertainty}
\label{appendix:chError}
In this appendix, we give the statistics of the errors caused by the channel uncertainty at the base station and user equipment. The errors caused by the channel uncertainty at the $g^{th}$ BS and the downlink user $k_g^d$ can be written as
\begin{align}
    \mathbf{e}_{g} & = \mathbf{\Delta}_{g,k_g^u} \mathbf{x}_{k_g^u} + \mathop{\sum^{G}\sum^{K_i^u}}_{i,j\neq k,g} \mathbf{\Delta}_{g,i_j^u} \mathbf{x}_{i_j^u} +  \sum_{j\neq g}^{G} \mathbf{\Delta}_{g,j} \mathbf{x}_j , \\
    \mathbf{e}_{k_g^d} & = \mathbf{\Delta}_{k_g^d,g} \mathbf{x}_{k_g^d} + \mathop{\sum^{G}\sum^{K_i^d}}_{i,j\neq k,g} \mathbf{\Delta}_{k_g^d,j} \mathbf{x}_{i_j^d} + \sum_{j=1}^{G} \sum_{i=1}^{K_j^u} \mathbf{\Delta}_{k_g^d,i_j^u} \mathbf{x}_{i_j^u} .
\end{align}
For any two specific nodes $A$ and $B$, we can calculate the covariance of the associated errors as
\begin{align}
    Cov \left( \mathbf{\Delta}_{B,A} \mathbf{x}_{A} \right) & = \mathbb{E} \left\{ \mathbf{\Delta}_{B,A} \mathbf{x}_{A} \mathbf{x}_{A}^H \mathbf{\Delta}_{B,A}^H \right\} \\
    & = \mathbb{E}_{\mathbf{\Delta}_{B,A}} \left\{ \mathbf{\Delta}_{B,A} \mathbf{T}_{A} \mathbf{\Delta}_{B,A}^H \right\}  \stackrel{(a)}{=} \tilde{\sigma}_{B,A}^2 tr \left( \mathbf{T}_{A} \right) \mathbf{I} \nonumber,
\end{align}
where $(a)$ comes from following derivation:
\begin{align}
        \mathbb{E}_{\mathbf{\Delta}_{B,A}} & \left\{ \left[ \mathbf{\Delta}_{B,A} \mathbf{T}_{A} \mathbf{\Delta}_{B,A}^H \right]_{m,n} \right\} \nonumber \\
        & =  \mathbb{E}_{\Delta} \left \{ \sum_{k=1}^{N_{\Delta}} \sum_{j=1}^{N_{\Delta}} \Delta _{mk} T_{kj} \Delta_{nj}^* \right \} \nonumber \\
        & = \sum_{k=1}^{N_{\Delta}} \sum_{j=1}^{N_{\Delta}} \mathbb{E}_{\Delta} \left \{ \Delta _{mk} T_{kj} \Delta_{mk}^* \right \} \delta_{mn} \delta_{kj} \\
        & \stackrel{(c)}{=} \sum_{k=1}^{N_{\Delta}} \mathbb{E}_{\Delta} \left| \Delta _{mk} \right|^2 T_{kk} \delta_{mn} = \tilde{\sigma}_{B,A}^2 tr\left( \mathbf{T}_{A} \right)\delta_{mn} \nonumber,
\end{align}
where $N_{\Delta}$ denotes the number of columns of matrix $\mathbf{\Delta}_{B,A}$, $T_{kj}$ denotes the element at the $k^{th}$ row and $j^{th}$ column of matrix $\mathbf{T}_{A}$, and $\Delta_{mk}$ denotes the element at the $m^{th}$ row and $k^{th}$ column of matrix $\mathbf{\Delta}_{B,A}$. $\delta$ represents the correlation coefficient, and $\delta_{kj} = 0, \forall \; k \neq j$, $\delta_{kk}=1, \forall \;k$, $\delta_{mn}=0,\forall \; m\neq n$, and $\delta_{mm}=1, \forall \;m$. Thus, the covariance matrices of the errors are given as
\begin{align}
    Cov & \left( \mathbf{e}_{g} \right) =  \underbrace{ \sum_{j=1}^{G} \sum_{i=1}^{K_j^u} \tilde{\sigma}_{g,i_j^u}^2 tr \left( \mathbf{T}_{i_j^u} \right) +  \sum_{j\neq g}^{G} \tilde{\sigma}_{g,j}^2 tr \left( \mathbf{T}_j \right)}_{\hat{\sigma}_g^2} \mathbf{I}_{M_{bs}} \\
    Cov & \left( \mathbf{e}_{k_g^d} \right) \\
    & =  \underbrace{ \sum_{j=1}^{G} \sum_{k=1}^{K_j^d} \tilde{\sigma}_{k_g^d,j}^2 tr \left( \mathbf{T}_{i_j^d} \right) + \sum_{j=1}^{G} \sum_{i=1}^{K_j^u} \tilde{\sigma}_{k_g^d,i_j^u}^2 tr \left( \mathbf{T}_{i_j^u} \right)}_{\hat{\sigma}_{k_g^d}^2} \mathbf{I}_{M_{ue}} \nonumber .
\end{align}

% ----------------------------------------------------------------------------------------------------------------------------------------------------------
\section{Proof of Lemma}
\label{appendix:lemma_proof}
The only difference between the two formulations (i.e., problems $(P.1)$ and $(P.2)$) is that the RSI power-related terms appear in different places, one is in the constraint condition, and the other is in the objective function. Since the RSI power is only related to the downlink precoders, this difference does not affect the solution to other variables. The Lagrange function of the two problems with respect to the downlink precoding matrices $\mathbf{V}_{k_g^d}$ can be given as
\begin{equation}
\label{eq:proof1}
    \mathcal{L}_{P_1} = \mathbf{\Sigma}_{L} \left ( \mathbf{V} \right ) +  \sum_{g=1}^{G} \nu_g \epsilon_{\text{rsi}, g}\left ( \mathbf{V}\right ) ,
\end{equation}
\begin{equation}
\label{eq:proof2}
    \mathcal{L}_{P_2} = \mathbf{\Sigma}_{L} \left ( \mathbf{V} \right ) +  \sum_{g=1}^{G} \varpi_{\text{rsi}, g} \left ( \epsilon_{\text{rsi}, g}\left ( \mathbf{V}\right ) - \bar{\epsilon}_{\text{rsi}, g} \right ), 
\end{equation}
where $\varpi_{\text{rsi}, g}$ is the Lagrange multiplier; $\mathbf{\Sigma}_{L} \left ( \mathbf{V} \right )$ denotes the objective and power constraint-related terms given as
\begin{equation}
    \begin{split}
        \mathbf{\Sigma}_{L} \left ( \mathbf{V} \right ) & = \sum_{g=1}^{G}\sum_{k=1}^{K_g^d} \varepsilon_{k_g^d}  \left ( \mathbf{V}  \right ) + \sum_{g=1}^{G}\sum_{k=1}^{K_g^u} \varepsilon_{k_g^u} \left ( \mathbf{V}  \right ) \\
        & \quad + \sum_{g=1}^{G}  \varpi_{g} \left ( \sum_{k=1}^{K_g^d} \alpha_{k_g^d}^2 tr \left ( \mathbf{V}_{k_g^d} \mathbf{V}_{k_g^d}^H \right ) - P_{bs} \right ) .
    \end{split}
\end{equation}
The optimal solutions to $\mathbf{V}_{k_g^d}$ are obtained by deriving the Lagrange function \eqref{eq:proof1} or \eqref{eq:proof2} with respect to $\mathbf{V}_{k_g^d}$ and set the derivatives to zero, which will be identical if $\nu_g = \varpi_{\text{rsi}, g}$. Thus, optimization problems $(P.1)$ and $(P.2)$ share the same solutions (i.e., $(P.1)$ and $(P.2)$ are equivalent) with $\nu_g = \varpi_{\text{rsi}, g}$. It should be noted that the optimal value of the Lagrange multiplier $\varpi_{\text{rsi}, g}$ depends on the tolerable RSI power we set.

% ----------------------------------------------------------------------------------------------------------------------------------------------------------
\section{Simplicities for MSE Expressions}
\label{appendix:simp_mse}
In this appendix, we give some simplicities used for MSE expressions.
\begin{equation}
    tr \left ( \mathbb{E} \left \{ \mathbf{s}_{k_g^d} \mathbf{s}_{k_g^d}^H \right \} \right ) = tr \left ( \mathbf{I}_{b_d} \right ) = b_d ;
\end{equation}
\begin{equation}
    tr \left ( \mathbb{E} \left \{ \mathbf{s}_{k_g^u} \mathbf{s}_{k_g^u}^H \right \} \right ) = tr \left ( \mathbf{I}_{b_u} \right ) = b_u;
\end{equation}
\begin{equation}
\begin{split}
    \mathbf{T}_{g} = \mathbb{E}\left \{ \mathbf{x}_g \mathbf{x}_g^H \right \} & = \mathbb{E}\left \{ \left ( \mathbf{V}_g \mathbf{A}_g \mathbf{s}_g + \mathbf{c}_g \right ) \left ( \mathbf{V}_g \mathbf{A}_g \mathbf{s}_g + \mathbf{c}_g \right )^H \right \} \\
    & = \mathbf{V}_g \mathbf{A}_g \mathbf{A}_g^H \mathbf{V}_g^H + \sigma_t^2 \mathcal{D} \left ( \mathbf{V}_g \mathbf{A}_g \mathbf{A}_g^H \mathbf{V}_g^H \right ) \\
    & =\sum_{k=1}^{K_g^d} \underbrace{\alpha_{k_g^d}^2 \left ( \mathbf{V}_{k_g^d} \mathbf{V}_{k_g^d}^H + \sigma_t^2 \mathcal{D} \left ( \mathbf{V}_{k_g^d} \mathbf{V}_{k_g^d}^H \right ) \right )}_{\mathbf{T}_{k_g^d}} ;
\end{split}
\end{equation}
\begin{equation}
\begin{split}
    \mathbf{T}_{k_g^u} & = \mathbb{E}\left \{ \mathbf{x}_{k_g^u} \mathbf{x}_{k_g^u}^H \right \} \\
    & = \mathbb{E}\left \{ \left ( \gamma_{k_g^u} \mathbf{V}_{k_g^u} \mathbf{s}_{k_g^u} + \mathbf{c}_{k_g^u} \right ) \left ( \gamma_{k_g^u} \mathbf{V}_{k_g^u} \mathbf{s}_{k_g^u} + \mathbf{c}_{k_g^u} \right )^H \right \} \\
    & = \gamma_{k_g^u}^2 \mathbf{V}_{k_g^u} \mathbf{V}_{k_g^u}^H + \sigma_t^2 \gamma_{k_g^u}^2 \mathcal{D} \left ( \mathbf{V}_{k_g^u} \mathbf{V}_{k_g^u}^H \right ) \\
    & = \gamma_{k_g^u}^2 \left ( \mathbf{V}_{k_g^u} \mathbf{V}_{k_g^u}^H + \sigma_t^2 \mathcal{D} \left ( \mathbf{V}_{k_g^u} \mathbf{V}_{k_g^u}^H \right ) \right ) ;
\end{split}
\end{equation}
\begin{equation}
\label{eq:Ckgd}
\begin{split}
    \mathbf{C}_{k_g^d} & = \mathbb{E} \left \{ \mathbf{y}_{k_g^d} \mathbf{y}_{k_g^d}^H \right \} \\
    & = \sum_{j=1}^{G} \sum_{i=1}^{K_j^d} \hat{\mathbf{H}}_{k_g^d,j} \mathbf{T}_{i_j^d} \hat{\mathbf{H}}_{k_g^d,j}^H + \sum_{j=1}^{G} \sum_{i=1}^{K_j^u} \hat{\mathbf{H}}_{k_g^d,i_j^u} \mathbf{T}_{i_j^u} \hat{\mathbf{H}}_{k_g^d,i_j^u}^H \\
    & \quad + \sigma_r^2 \mathcal{D} \bigg ( \sum_{j=1}^{G} \sum_{i=1}^{K_j^d} \hat{\mathbf{H}}_{k_g^d,j} \mathbf{T}_{i_j^d} \hat{\mathbf{H}}_{k_g^d,j}^H \\
    & \quad + \sum_{j=1}^{G} \sum_{i=1}^{K_j^u} \hat{\mathbf{H}}_{k_g^d,i_j^u} \mathbf{T}_{i_j^u} \hat{\mathbf{H}}_{k_g^d,i_j^u}^H \bigg ) + (\sigma_{k_g^d}^2 + \tilde{\sigma}_{k_g^d}^2)\mathbf{I}_{M_{ue}} .
\end{split}
\end{equation}
\begin{equation}
\label{eq:Cg}
\begin{split}
    \mathbf{C}_{g} & = \mathbb{E} \left \{ \mathbf{y}_{g} \mathbf{y}_{g}^H \right \} \\
    & = \sum_{j=1}^{G} \sum_{i=1}^{K_j^u} \hat{\mathbf{H}}_{g,i_j^u} \mathbf{T}_{i_j^u} \hat{\mathbf{H}}_{g,i_j^u}^H + \sum_{j=1}^{G} \hat{\mathbf{H}}_{g,j} \mathbf{T}_j \hat{\mathbf{H}}_{g,j}^H  \\
    & \quad + \sigma_r^2 \mathcal{D} \bigg ( \sum_{j=1}^{G} \sum_{i=1}^{K_j^u} \hat{\mathbf{H}}_{g,i_j^u} \mathbf{T}_{i_j^u} \hat{\mathbf{H}}_{g,i_j^u}^H \\
    & \quad + \sum_{j=1}^{G} \hat{\mathbf{H}}_{g,j} \mathbf{T}_j \hat{\mathbf{H}}_{g,j}^H \bigg ) + (\sigma_{g}^2 + \tilde{\sigma}_{g}^2) \mathbf{I}_{M_{bs}}.
\end{split}
\end{equation}

% ----------------------------------------------------------------------------------------------------------------------------------------------------------
\section{NSP-based Method}
\label{appendix:NSP-MMSE}
In this appendix, we will explain how to apply the null-space projection to enhance the ASIC capability of precoders. The processing consists of two steps: 1) obtain the desired precoder; 2) project the precoder to the null-space of the SI channel. The first step could be done using any off-the-shelf precoding techniques, and we use the MWSR beamforming here for consistency, which can be obtained as in \cite{paulaWsr}. Let $ \mathbf{V}_{k_g^d}^{\dagger} $ denote the desired precoder obtained by the MWSR algorithm, then we project the desired precoder as $ \mathbf{\Gamma}_{g, D} \mathbf{\Gamma}_{g, D}^H \mathbf{V}_{k_g^d}^{\dagger}$ \cite{tbSic}, where $ \mathbf{\Gamma}_{g, D}$ is the chosen subspace that spans of the eigenvectors associated with the $D$ smallest eigenvalues of $\mathbf{H}_{g,g} \mathbf{H}_{g,g}^H + \kappa_{bs} \mathcal{D} \left ( \mathbf{H}_{g,g} \mathbf{H}_{g,g}^H \right )$.

% ----------------------------------------------------------------------------------------------------------------------------------------------------------
\section{Functions for Explicit Expressions}
\label{appendix:functions}
In this appendix, we will define two functions for explicit expressions.
\begin{align}
    \mathcal{F}_1 & \left ( \mathbf{Y}, \mathbf{X} \right )  = \mathbf{Y} \mathbf{X} \mathbf{X}^H \mathbf{Y}^H + \sigma_t^2 \mathcal{D} \left ( \mathbf{Y} \mathbf{X} \mathbf{X}^H \mathbf{Y}^H \right ) \nonumber \\
    & \quad + \sigma_r^2 \mathbf{Y} \mathcal{D} \left ( \mathbf{X} \mathbf{X}^H \right ) \mathbf{Y}^H + \sigma_r^2 \sigma_t^2 \mathcal{D} \left (  \mathbf{Y} \mathcal{D} \left ( \mathbf{X} \mathbf{X}^H \right ) \mathbf{Y}^H \right )
\end{align}
\begin{align}
    \mathcal{F}_2 & \left ( \mathbf{Z}, \mathbf{Y}, \mathbf{X} \right ) = tr \large ( \mathbf{Z}^H\mathbf{Y} \left ( \mathbf{X} \mathbf{X}^H + \sigma_t^2 \mathcal{D} \left (  \mathbf{X} \mathbf{X}^H \right ) \right ) \mathbf{Y}^H \mathbf{Z} \nonumber \\
    & \quad + \sigma_r^2 \mathbf{Z}^H \mathcal{D} \left (\mathbf{Y} \left ( \mathbf{X} \mathbf{X}^H + \sigma_t^2 \mathcal{D} \left (  \mathbf{X} \mathbf{X}^H \right ) \right ) \mathbf{Y}^H \right ) \mathbf{Z}  \large )
\end{align}

\ifCLASSOPTIONcaptionsoff
  \newpage
\fi

% ----------------------------------------------------------------------------------------------------------------------------------------------------------

\vfill

\end{document}